\documentclass[11pt]{article}
\usepackage{a4wide}
\usepackage{amsthm,amsfonts,amsbsy,graphicx,calrsfs,amssymb}
\usepackage{amsmath}
\usepackage{times}
\usepackage{tikz}
\usepackage{xcolor,colortbl}
\usepackage{rotating}
\usepackage{hyperref}
\usetikzlibrary{arrows}
\usepackage[nomarkers,figuresonly,nofiglist]{endfloat}

\allowdisplaybreaks
\theoremstyle{plain}
\newtheorem{theorem}{Theorem}[section]
\newtheorem{lemma}{Lemma}[section]
\newtheorem{corollary}{Corollary}[section]

\theoremstyle{definition}
\newtheorem{definition}{Definition}[section]
\newtheorem{example}{Example}[section]
\newtheorem{remark}{Remark}[section]

\numberwithin{equation}{section}

\newcommand{\Linn}{L^{0}(\mathbb{I}^{m})}
\newcommand{\Lina}{L^{0}(\mathbb{I}^{m_1})}
\newcommand{\Linb}{L^{0}(\mathbb{I}^{m_2})}

\newcommand{\R}{{\mathbb R}}

\newcommand{\CC}{{\mathcal C}}

\newcommand{\eee}{{\bf e}}

\newcommand{\uuu}{{\bf u}}
\newcommand{\vvv}{{\bf v}}

\newcommand{\ein}{{\bf 1}}
\newcommand{\nul}{{\bf 0}}
\newcommand{\leb}{\boldsymbol{\lambda}}

\begin{document}

\title{Dissimilarity functions for rank-invariant hierarchical clustering of continuous variables
}
\author{
Sebastian Fuchs
\thanks{Department for Mathematics, Universit{\"a}t Salzburg, Salzburg, Austria.
E-mail: \texttt{sebastian.fuchs@sbg.ac.at} (Corresponding author).}
\and
F. Marta L. Di Lascio
\thanks{Faculty of Economics and Management, Free University of Bozen-Bolzano, Bozen-Bolzano, Italy.
E-mail: \texttt{marta.dilascio@unibz.it}}
\and
Fabrizio Durante
\thanks{Dipartimento di Scienze dell'Economia, Universit\`a del Salento, Lecce, Italy.
E-mail: \texttt{fabrizio.durante@unisalento.it}}}

\date{}

\maketitle

\begin{abstract}
\noindent
A theoretical framework is presented for a (copula-based) notion of dissimilarity between continuous random vectors and its main properties are studied. The proposed dissimilarity assigns the smallest value to a pair of random vectors that are comonotonic.
Various properties of this dissimilarity are studied, with special attention to those that are prone to the hierarchical agglomerative methods, such as reducibility. 
Some insights are provided for the use of such a measure in clustering algorithms and a simulation study is presented. 
Real case studies illustrate the main features of the whole methodology. \\

\noindent\textit{Keywords:} Comonotonicity, Copula, Cluster analysis, Dissimilarity, Stochastic dependence.\\
\end{abstract}

\section{Introduction}\label{S:intro}

A clustering method aims at visualizing the relationships among objects, say $\mathbf{x}_1,\dots,\mathbf{x}_m$, so that one can understand their main features. In particular, hierarchical clustering methods represent the relationships between $m$ objects based on their (pairwise) dissimilarities in the form of a tree, where each leaf corresponds to one of the original objects, namely $\mathbf{x}_i$, and each interior node represents a subset or cluster of objects. Moreover, agglomerative hierarchical clustering algorithms build trees in a bottom-up approach, beginning with $n$ singleton clusters of the form $\{\mathbf{x}_i\}$ and, then, merging the two closest clusters at each stage, until only one cluster remains. The resulting binary tree formed by this process can also provide an intuitive graphical representation in terms of a dendrogram. See, for instance, \cite{Eve11,Has09,Hen16}.

According to the different applications, $\mathbf{x}_1,\dots,\mathbf{x}_m$ may have various nature and interpretation.  Here, we suppose that they are sample data from the continuous random variables (r.v.'s hereafter) $X_1,\dots,X_m$ defined on the same probability space. Moreover, we consider hierarchical clustering algorithms that can visualize a specific kind of association (similarity) among the r.v.'s. We recall that measures of association capture the many facets of dependence relationships, from classical linear correlation coefficients to indices for detecting concordance, tail dependence, radial symmetry, etc. In their study, the contribution of copula methods to describe (mainly, continuous) r.v.'s has been largely recognized (see, e.g., \cite{DurSem16,Gib21,Groetal14,Joe15,MaiSch12,Nel06,Saletal07} and references therein).

In the literature, the use of (copula-based) hierarchical clustering procedures is mainly two-fold. First, such algorithms have been employed to 
guide the process of model building (and selection), especially in high dimensions. For instance, cluster algorithms have been used in \cite{Goretal17} for the identification of a nested Archimedean structure, and in \cite{Czaetal12,Disetal13} for the determination of a vine copula model, among others. Moreover, a procedure is illustrated in \cite{CotGen15} for selecting the tree structure of a risk aggregation model by combining hierarchical clustering techniques with a distance metric based on Kendall's tau. Finally, in \cite{Per19}, an iterative algorithm is proposed to group variables into clusters with exchangeable dependence.

Second, hierarchical clustering procedures have been used to detect comovements of r.v.'s (especially, in time series). In financial time series, for instance, these methods start with the use of (Pearson) correlation coefficient and some of its variants (see, for instance, \cite{Bonetal04,Eve11}) and, then, benefit from the copula approach especially when the detection of extreme dependence is of interest \cite{DeLZuc11,DeLZuc17b,DeLZuc17,DurPapTor14ADAC,DurPapTor15StatPap}. 
Related clustering methods can be built from other measures of association/concordance \cite{Bon19,DiLDurPap17RBN}, mutual information \cite{Koj04}, as well as from a dissimilarity derived from the empirical copula \cite{Disetal17,Koj10}.
Notice that the above procedures differ from model-based clustering techniques, which aim to group observations from the same subpopulation of a multivariate mixture distribution (see, e.g., \cite{KosKar16,Mar17}) and from the CoClust algorithm, which aims to groups observations according to the multivariate dependence structure of the data generating copula (see \cite{Dilascio16}).

Motivated by the interest in clustering methods of (agglomerative) hierarchical type, we introduce and formalize a notion of dissimilarity between two subsets of r.v.'s. Such a  dissimilarity measure assigns the smallest value to two subsets of r.v.'s that are pairwise comonotonic (see, e.g., \cite{Dhaetal02a,KocDes11,PucSca10,PucWan15}). Moreover, this measure is of probabilistic nature, i.e. it depends on the joint probability distribution function of the involved variables, and is copula-based, i.e. it is invariant under monotonically increasing transformations of the involved r.v.'s.

Specifically, we investigate whether a dissimilarity measure can satisfy some desirable theoretical properties (for example, reducibility) that are satisfied by some classical clustering methods based on Euclidean distances.

The paper is organized as follows. First, we define the general framework where our dissimilarity concept is build up (Section \ref{sec:frame}) and, in particular, we introduce some desirable properties that a dissimilarity may satisfy with particular emphasis on those properties that are prone to the hierarchical agglomerative methods. In Section \ref{sec:example}, we consider and compare to each other various examples of dissimilarity mappings, including those methods based on linkage functions. Moreover, in Section \ref{sec:algorithm} we show how the dissimilarity can be used to detect various kinds of stochastic dependence of a random vector. From the computational side, we hence provide a simulation study in order to show how algorithms based on linkage and pairwise dissimilarities work in a finite sample and discuss their advantages and disadvantages (Section \ref{subsec:simula}). Real case studies illustrate the whole methodology (Section \ref{sec:appli}). Remarkably, we also present a case when the use of novel dissimilarity measures (not based on linkage functions) may be beneficial in the detection of global dependencies. Section \ref{sec:concl} summarizes the main findings.

\section{The framework}\label{sec:frame}

Throughout this manuscript, we consider r.v.'s  defined on the same probability space $(\Omega,\mathcal{F},\mathbb{P})$. We denote by $F_X$ the probability distribution function of a r.v.~$X$ and by $F_{X}^{(-1)}$ its generalized inverse (see, e.g., \cite{EmbHof13}). We recall that, for a continuous r.v.~$X$, the composition $F_X \circ X$ is uniformly distributed on $\mathbb{I}:=[0,1]$ (see, e.g., \cite{DurSem16}). 
Let $\vec{\mathbf X}$ be a continuous random vector, i.e. $\vec{\mathbf X}=(X_1,\dots,X_m)$. In view of Sklar's Theorem, the copula of $\vec{\mathbf X}$ is the distribution function of $(F_{X_1}(X_1), \dots, F_{X_{m}}(X_{m}))$. 

For the sake of completeness, we recall that the copula $M$ is defined by $M(u_1\dots,u_m)=\min\{u_1,\dots,u_m\}$ for all $u_1,\dots,u_m$ in $\mathbb{I}$. In particular, two continuous r.v.'s  $X$ and $Y$ are said to be \emph{comonotonic} if their copula is equal to $M$. Comonotonicity of r.v.'s  can also be expressed in one of the following equivalent ways:
\begin{itemize}
\item[(a)] $(X,Y)\stackrel{d}{=} (F_X^{(-1)}(U),F_Y^{(-1)}(U))$, where $U$ is a uniform random variable;
\item[(b)] there exists a r.v.~$Z$ such that $(X,Y)\stackrel{d}{=} (f_1(Z),f_2(Z))$ for some increasing functions $f_1,f_2$.
\end{itemize}
See, e.g., \cite{Dhaetal02a,PucSca10}.

\bigskip
In the following, let $ m \geq 2 $ be an integer which will be kept fixed. We consider a (finite) set $\mathcal{X}=\{X_1,\dots,X_m\}$ of continuous r.v.'s. Any subset of $\mathcal{X}$ will be denoted by upper-case black-board letters, e.g. $\mathbb{X}$. 
Let $\mathcal{P}_0(\mathcal{X})$ denote the set of all non--empty subsets of $\mathcal{X}$.

Given a subset $\mathbb{X}=\{X_1, \dots, X_k\}\subset \mathcal{X}$ composed of $k$ r.v.'s, we indicate by $\vec{\mathbb{X}}$ a vector representation of 
$\mathbb{X}$, i.e. a $k$--dimensional random vector whose coordinates are distinct elements from $\mathbb{X}$. Clearly, the vector representation of any $\mathbb{X}$ need not be unique.

Here, we aim at quantifying how two non-empty subsets of $\mathcal{X}$
(not necessarily equal in cardinality) are similar or, analogously, how we can define a suitable dissimilarity index between them. The main properties that this index should satisfy are illustrated in the following.

A \textit{dissimilarity index} is a mapping
$\widetilde{d}$ that assigns to every pair $(\mathbb{X},\mathbb{Y}) \in \mathcal{P}_0(\mathcal{X})\times \mathcal{P}_0(\mathcal{X})$ a value in $[0,+\infty[$ with the following properties:
\begin{itemize}
\item[\textnormal{($\widetilde{A1}$)}]
$ \widetilde{d}(\mathbb{X},\mathbb{Y})=0$
holds for all $\mathbb{X}=\{X_1,\dots,X_{m_1}\}, \mathbb{Y}=\{Y_1,\dots,Y_{m_2}\} \in \mathcal{P}_0(\mathcal{X})$ such that the r.v.'s $X_1,\dots,X_{m_1},Y_1,\dots,Y_{m_2}$ are all pairwise comonotonic.

\item[\textnormal{($\widetilde{A2}$)}]
$ \widetilde{d}(\mathbb{X},\mathbb{Y})
	= \widetilde{d}(\mathbb{Y},\mathbb{X}) $
holds for all $\mathbb{X}, \mathbb{Y} \in \mathcal{P}_0(\mathcal{X})$.

\item[\textnormal{($\widetilde{A3}$)}]
The identity
$$
  \widetilde{d}(\mathbb{X},\mathbb{Y})
	= \widetilde{d}(\mathbb{X}_1,\mathbb{Y}_1)
$$
holds for all $\mathbb{X},\mathbb{Y},\mathbb{X}_1,\mathbb{Y}_1 \in \mathcal{P}_0(\mathcal{X})$ such that there exist some vector representations of $\mathbb{X},\mathbb{Y},\mathbb{X}_1,\mathbb{Y}_1 $ for which it holds $(\vec{\mathbb X},\vec{\mathbb Y}) \stackrel{d}{=} (\vec{\mathbb X}_1,\vec{\mathbb Y}_1)$.
\item[\textnormal{($\widetilde{A4}$)}]
The identity
$$
  \widetilde{d} \big( \mathbb{X}, \mathbb{Y} \big)
	= \widetilde{d} \big( \{T_1(X_1), \dots, T_{m_1}(X_{m_1})\}, \{Y_1, \dots, Y_{m_2}\} \big)
$$
holds for all $\mathbb{X} = \{X_1, \dots, X_{m_1}\}, \mathbb{Y} = \{Y_1, \dots, Y_{m_2}\} \in \mathcal{P}_0(\mathcal{X})$
and every set of strictly increasing transformations $\{T_1, \dots, T_{m_1}\}$.
\end{itemize}

Condition ($\widetilde{A1}$) implies that $\widetilde{d}(\{X\},\{Y\})=0$ when $X$ and $Y$ are comonotonic. Thus, roughly speaking, this index quantifies the closeness of the joint distribution of $(X,Y)$ to the upper bound of the related Fr\'echet class (see, also, \cite{CifReg17,GDA16} for an historical overview). In general, by extending slightly the notation in \cite{PucSca10}, the dissimilarity index between two subsets $\mathbb{X}$ and $\mathbb{Y}$ is minimal when any vector representation of $\mathbb{X}$ is strongly comonotonic with any vector representations of $\mathbb{Y}$.

Condition ($\widetilde{A2}$) expresses a natural symmetry property of $\widetilde{d}$. Condition ($\widetilde{A3}$), instead, states that the dissimilarity index is law-invariant, as the various notions of association considered in the literature. At this point note that conditions
$\vec{\mathbb X} \stackrel{d}{=} \vec{\mathbb X}_1$ and
$\vec{\mathbb Y} \stackrel{d}{=} \vec{\mathbb Y}_1$ are essential since the above law-invariance is limited to pairs of subsets of fixed size and fixed distribution; for instance, $\widetilde{d} \big( \{X_1, X_{2}\}, \{Y_1, Y_{2}\} \big)$ can be different from $\widetilde{d} \big( \{X_1\}, \{X_{2}, Y_1, Y_{2}\} \big)$.

Condition ($\widetilde{A4}$), together with ($\widetilde{A2}$), states that the dissimilarity index between $\mathbb{X}$ and $\mathbb{Y}$ is invariant under strictly increasing transformations of their respective elements and, hence, it does not depend on the univariate distribution functions of the involved r.v.'s. Therefore, since it is precisely the copula which captures those properties of the joint distribution which are invariant under strictly increasing transformations, the dissimilarity index is a copula-based concept.

\begin{remark}
Notice that the previous conditions, especially ($\widetilde{A1}$), distinguishes the proposed methodology with other methods that enable the detection of all types of functional dependencies among variables (see, e.g., \cite{Koj04}). Indeed, the proposed dissimilarity index essentially aims at finding monotonic functional dependencies among the involved r.v.'s.
\end{remark}

Due to the above stated rank-invariant property of the dissimilarity, without loss of generality, we can introduce a dissimilarity index by considering r.v.'s that are uniformly distributed on $\mathbb{I}$ (i.e. working directly in the class of copulas). Before formalizing this aspect (see next Theorem \ref{diss.fct.index.relation}), we need some preliminary definitions.

Let $\Linn$ denote the space of all $m$-dimensional random vectors with uniform margins on $\mathbb{I}$.

\begin{definition} \label{diss.fct.local}
For all $ m_1, m_2 \in \mathbb{N} $ with $ 2 \le m_1+m_2\le m $,
$$
d^{m_1,m_2} \colon \Lina\times\Linb\to [0,+\infty[
$$
is called a $(m_1,m_2)$-\emph{dissimilarity function} if it satisfies the following properties:
\begin{itemize}
\item[(1)] For every $ (\vec{\mathbb{X}}, \vec{\mathbb{Y}}) \in \Lina \times \Linb $
\begin{equation} \label{local.diss.1}
  d^{m_1,m_2}(\vec{\mathbb{X}},\vec{\mathbb{Y}})= 0
\end{equation}
when the copula of $(\vec{\mathbb{X}},\vec{\mathbb{Y}})$ is equal to the comonotonicity copula $M$.
\item[(2)]
For every $(\vec{\mathbb{X}}, \vec{\mathbb{Y}}) \in \Lina \times \Linb$
\begin{equation} \label{local.diss.2}
d^{m_1,m_2}(\vec{\mathbb{X}},\vec{\mathbb{Y}})
	= d^{m_1,m_2}(\sigma_1(\vec{\mathbb{X}}),\sigma_2(\vec{\mathbb{Y}}))
\end{equation}
holds for all $\sigma_1$ and $\sigma_2$ permuting the coordinates of a vector from $\mathbb{I}^{m_1}$ and $\mathbb{I}^{m_2}$, respectively.

\item[(3)]
The identity
\begin{equation} \label{local.diss.3}
  d^{m_1,m_2} (\vec{\mathbb{X}},\vec{\mathbb{Y}})
	= d^{m_1,m_2} (\vec{\mathbb{X}}_1,\vec{\mathbb{Y}}_1)
\end{equation}
holds for all $(\vec{\mathbb{X}}, \vec{\mathbb{Y}}), (\vec{\mathbb{X}}_1,\vec{\mathbb{Y}}_1) \in \Lina \times \Linb$
with $(\vec{\mathbb X},\vec{\mathbb Y}) \stackrel{d}{=} (\vec{\mathbb X}_1,\vec{\mathbb Y}_1)$.
\end{itemize}
\end{definition}

Properties \eqref{local.diss.1} and \eqref{local.diss.3}
are direct translations of ($\widetilde{A1}$) and ($\widetilde{A3}$) in a copula setting; the latter property \eqref{local.diss.3} states that the dissimilarity between two vectors is law-invariant (and hence only depends on the copula involved).

Property \eqref{local.diss.2} states that the dissimilarity between two vectors does not change when the components of each vector are permuted. Roughly speaking, the dissimilarity does depend on the components of a random vector, but not on the order they are considered (so, it is a property about sets not vectors).

All the dissimilarity functions can be glued together into the following concept.

\begin{definition} \label{diss.fct.global}
An \emph{extended dissimilarity function} (of degree $m$) is a map
\begin{equation*}
  d: \bigcup_{2 \leq m_1+m_2 \leq m} \Lina\times\Linb \to [0,+\infty[
\end{equation*}
whose restriction
$d^{m_1,m_2} := d|_{\Lina\times\Linb} $ to $\Lina\times\Linb$ is a $(m_1,m_2)$-\emph{dissimilarity function} for all $ m_1, m_2 \in \mathbb{N} $ with $ 2 \le m_1+m_2\le m $, such that, for every $ (\vec{\mathbb{X}}, \vec{\mathbb{Y}}) \in \Lina \times \Linb$,
\begin{equation}\label{eq:simmetry}
d^{m_1,m_2}(\vec{\mathbb{X}},\vec{\mathbb{Y}})
=d^{m_2,m_1}(\vec{\mathbb{Y}},\vec{\mathbb{X}})
\end{equation}
holds.
\end{definition}

Roughly speaking, an extended dissimilarity function allows to assign a degree of dissimilarity to any pair of random vectors, regardless of the respective dimension. The condition given by \eqref{eq:simmetry} simply ensures that the dissimilarity has some natural symmetry related to ($\widetilde{A2}$).

Theorem \ref{diss.fct.index.relation} below demonstrates that the notion of extended dissimilarity function is consistent with the notion of dissimilarity index.

\begin{theorem}\label{diss.fct.index.relation}
The following statements hold:
\begin{enumerate}
\item[\textnormal{(a)}]
Let $d$ be an extended dissimilarity function. Then, the mapping $\tilde{d}\colon \mathcal{P}_0(\mathcal{X})\times \mathcal{P}_0(\mathcal{X})\to [0,+\infty[$ given, for every $\mathbb{X}=\{X_1, \dots, X_{m_1}\}$ and $\mathbb{Y}= \{Y_1, \dots, Y_{m_2}\}$, by
\begin{equation}\label{eq:d-tilde}
\widetilde{d} \big( \mathbb{X},\mathbb{Y} \big)
:=  d^{m_1,m_2} \big( (F_{X_1}(X_1), \dots, F_{X_{m_1}}(X_{m_1})), (F_{Y_1}(Y_1), \dots, F_{Y_{m_2}}(Y_{m_2})) \big)
\end{equation}
is a dissimilarity index.
\item[\textnormal{(b)}]
Let $\tilde{d}$ be a dissimilarity index.
Then, the map
\begin{equation*}
  d: \bigcup_{2 \leq m_1+m_2 \leq m} \Lina\times\Linb \to [0,+\infty[
\end{equation*}
given by
\begin{equation*}
  d \big( (X_1, \dots, X_{m_1}), (Y_1, \dots, Y_{m_2}) \big)
	:= \tilde{d} \big( \{X_1, \dots, X_{m_1}\}, \{Y_1, \dots, Y_{m_2}\} \big)
\end{equation*}
is an extended dissimilarity function.
\end{enumerate}
\end{theorem}
\begin{proof}
Consider assertion (a). First, consider that \eqref{eq:d-tilde} is well-defined since, in view of property \eqref{local.diss.2}, it does not depend on the specific vector representation of $\{X_1,\dots,X_{m_1}\}$ and $\{Y_1,\dots,Y_{m_2}\}$. 
Note that, for continuous distributions functions $F_X$, the composition $F_X \circ X$ is uniformly distributed on $\mathbb{I}$. Therefore, for every set of r.v.'s $\{X_1, \dots, X_{m_1}\}$,
the transformed vector satisfies $(F_{X_1}(X_1), \dots, F_{X_{m_1}}(X_{m_1})) \in \Lina$.
Moreover,
($\widetilde{A1}$), ($\widetilde{A2}$) and ($\widetilde{A3}$) are direct consequences of properties
\eqref{local.diss.1}, \eqref{eq:simmetry} and \eqref{local.diss.3} of Definitions \ref{diss.fct.local} and \ref{diss.fct.global}.
Finally, ($\widetilde{A4}$) follows from \eqref{local.diss.3} and the
fact that $F_{T \circ X} (T \circ X) = F_{X} (X)$ for every strictly increasing function $T$.

Assertion (2) is straightforward.
\end{proof}

It follows from the previous result that, from now on,
we can express the dissimilarity index in terms of suitable properties of the extended dissimilarity function $d$.

Before introducing basic examples of such functions, we present here some additional desirable properties they may satisfy. First, we present some \emph{local} properties, i.e. properties that are satisfied by the restriction $d^{m_1,m_2}$ of the extended dissimilarity function $d$ for any possible choice of $m_1,m_2$ with $2 \leq m_1+m_2 \leq m$.

\begin{itemize}
\item[(L1)] Monotonicity with respect to lower orthant order\\
For all $2\le m_1+m_2\le m$
and for all
$ (\vec{\mathbb{X}}, \vec{\mathbb{Y}}), (\vec{\mathbb{X'}}, \vec{\mathbb{Y'}}) \in \Lina \times \Linb $, \linebreak
$(\vec{\mathbb{X}},\vec{\mathbb{Y}})\preceq_{lo}(\vec{\mathbb{X'}},\vec{\mathbb{Y'}})$
in the lower orthant order implies
$ d^{m_1,m_2} (\vec{\mathbb{X'}},\vec{\mathbb{Y'}})
	\le d^{m_1,m_2} (\vec{\mathbb{X}},\vec{\mathbb{Y}}) $.

\item[(L1c)] Monotonicity with respect to concordance order\\
For all $2\le m_1+m_2\le m$ and for all
$ (\vec{\mathbb{X}}, \vec{\mathbb{Y}}), (\vec{\mathbb{X'}}, \vec{\mathbb{Y'}}) \in \Lina \times \Linb $, \linebreak
$(\vec{\mathbb{X}},\vec{\mathbb{Y}})\preceq_C(\vec{\mathbb{X'}},\vec{\mathbb{Y'}})$
in the concordance order implies
$ d^{m_1,m_2} (\vec{\mathbb{X'}},\vec{\mathbb{Y'}})
	\le d^{m_1,m_2} (\vec{\mathbb{X}},\vec{\mathbb{Y}}) $.

\item[(L2)]  Rotation invariance\\
For all $2\le m_1+m_2\le m$
and for every $ (\vec{\mathbb{X}}, \vec{\mathbb{Y}}) \in \Lina \times \Linb $, it holds \linebreak
$ d^{m_1,m_2}(\vec{\mathbb{X}},\vec{\mathbb{Y}})
	= d^{m_1,m_2}(\mathbf{1}_{m_1}-\vec{\mathbb{X}},\mathbf{1}_{m_2}-\vec{\mathbb{Y}}) $, where $\mathbf{1}_n$ is a vector with all $n$ components equal to $1$.

\item[(L3)] Continuity\\
For all $2\le m_1+m_2\le m$,
any sequence
$ \{ \vec{\mathbb{Z}}_k=(\vec{\mathbb{X}},\vec{\mathbb{Y}})_k \}_{k\in \mathbb{N}} \subseteq \Lina \times \Linb $
and any vector $ \vec{\mathbb{Z}}=(\vec{\mathbb{X}}, \vec{\mathbb{Y}}) \in \Lina \times \Linb$,
if $\vec{\mathbb{Z}}_k$ weakly converges to $\vec{\mathbb{Z}}$ (as $k$ tends to $+\infty$), then
$ \lim_{k \to \infty}  d^{m_1,m_2}(\vec{\mathbb{X}},\vec{\mathbb{Y}})_k
= d^{m_1,m_2}(\vec{\mathbb{X}},\vec{\mathbb{Y}})$.
\end{itemize}

Property (L1) (respectively, (L1c)) implies that the dissimilarity degree is decreasing with respect to lower orthant (also called PLOD) order (respectively, concordance order).
For the definitions of these orderings see, for instance, \cite{Nel06,MueSto02}.
Since the upper bound of a random vector in the lower orthant (respectively, concordance) order is given by the comonotonic case, this property simply means that the dissimilarity degree tends to vanish as soon as one is approaching the comonotonic case. Notice that, in the bivariate case, lower orthant and concordance order coincide, while in higher dimensions  concordance order implies lower orthant order, but not vice versa (see, e.g., \cite{Joe90,MueSca00}).
Consequently, monotonicity with respect to lower orthant order (L1) implies monotonicity with respect to concordance order (L1c).

Property (L2) expresses the invariance of the dissimilarity degree with respect to the total reflection of the involved random vectors. The practical aspect of this property is that a change of sign in all the r.v.'s does not influence the clustering output.
Notice that this property may not be desirable when the dissimilarity degree should distinguish lower and upper tail behaviour of random vectors (see Section \ref{subsec:TDC} and, in particular, Remark \ref{TDC.rem}).

Property (L3) ensures that the dissimilarity degree is continuous with respect to weak convergence.
This latter property is usually required, for instance, for various measures of concordance (see, e.g., \cite{Fuc16Demo,Sca84,Tay16})
but it does not apply to the tail dependence coefficient (see Section \ref{subsec:TDC}).

\bigskip
Now, we provide some global properties of an extended dissimilarity function $d$ that connect the values of the dissimilarity at a given dimension, say $m_1+m_2$, with the values that it assumes at lower (respectively, higher) dimensions:

\begin{itemize}
\item[(G1)] Reducibility\\
For all $3\le m_{1}+m_{2}+m_3\le m$ and for every $ (\vec{\mathbb{X'}}, \vec{\mathbb{X''}}, \vec{\mathbb{Y}}) \in L^{0}(\mathbb{I}^{m_{1}}) \times L^{0}(\mathbb{I}^{m_{2}}) \times L^{0}(\mathbb{I}^{m_{3}}) $ such that $\vec{\mathbb{X'}}$, $\vec{\mathbb{X''}}$, and $\vec{\mathbb{Y}}$ are pairwise disjoint, if
\begin{equation}\label{eq:reducibility1}
d^{m_{1},m_{2}} (\vec{\mathbb{X'}},\vec{\mathbb{X''}})
\leq \min\big\{ d^{m_{1},m_3} (\vec{\mathbb{X'}},\vec{\mathbb{Y}}), d^{m_{2},m_3} (\vec{\mathbb{X''}},\vec{\mathbb{Y}}) \big\},
\end{equation}
then the inequality
\begin{equation}\label{eq:reducibility2}
  \min\big\{ d^{m_{1},m_3} (\vec{\mathbb{X'}},\vec{\mathbb{Y}}), d^{m_{2},m_3} (\vec{\mathbb{X''}},\vec{\mathbb{Y}}) \big\}
  \leq d^{m_{1}+m_{2},m_3} (\vec{\mathbb{X}},\vec{\mathbb{Y}})
\end{equation}
holds,
where $ \mathbb{X} := \mathbb{X'} \cup \mathbb{X''} $.

\item[(G1s)] Strict reducibility\\
(G1) holds with \eqref{eq:reducibility2} being strict
for some $3\le m_{1}+m_{2}+m_3\le m$ and at least one $ (\vec{\mathbb{X'}}, \vec{\mathbb{X''}}, \vec{\mathbb{Y}}) \in L^{0}(\mathbb{I}^{m_{1}}) \times L^{0}(\mathbb{I}^{m_{2}}) \times L^{0}(\mathbb{I}^{m_{3}}) $,
where $ \mathbb{X} := \mathbb{X'} \cup \mathbb{X''} $
and $ \mathbb{X'} $, $\mathbb{X''} $ and $\mathbb{Y} $ are pairwise disjoint.

\item[(G2)] Comonotonic invariance\\
For all $3\le m_1+m_2\le m$ with $ 2 \le m_1 $, the identity
$$
  d^{m_1,m_2} (\vec{\mathbb{X}},\vec{\mathbb{Y}})
	= d^{m_1-1,m_2} (\vec{\mathbb{X'}},\vec{\mathbb{Y}})
$$
holds whenever $\vec{\mathbb{X'}} \in L^{0}(\mathbb{I}^{m_1-1}) $ and
$ (\vec{\mathbb{X}}, \vec{\mathbb{Y}}) \in \Lina \times \Linb$ are random vectors
such that $\mathbb{X'}\cup \{X\}=\mathbb{X}$, where $ X \in \mathbb{X} $ is comonotonic with at least one element of $\mathbb{X'}$.
\end{itemize}

All these properties have an intuitive stochastic interpretation.
Property (G1) is usually referred to as {\it reducibility property} (see, e.g.,  \cite{Koj10}). It guarantees that the dissimilarity degree between two random vectors $\vec{\mathbb{X}}$ and $\vec{\mathbb{Y}}$ is larger than the dissimilarity degree between $\vec{\mathbb{Y}}$ and (at least) a subvector of $\vec{\mathbb{X}}$. Roughly speaking, increasing the diversity inside each group decreases the similarity between the groups.

\begin{example}
Given three r.v.'s $X',X'',Y$, property (G1) ensures that, if \eqref{eq:reducibility1} holds, i.e. $(X',X'')$ is the most similar pair among $(X',X'')$, $(X',Y)$ and $(X'',Y)$, then
$$
d^{1,1} (X',X'') \le
d^{2,1} ((X',X''),Y).
$$
\end{example}

The related property (G1s) says, furthermore, that there exist specific dependence structures such that the dissimilarity degree between $\vec{\mathbb{X}}$ and $\vec{\mathbb{Y}}$ is strictly larger than the dissimilarity degree between $\vec{\mathbb{Y}}$ and a subvector of $\vec{\mathbb{X}}$.
Clearly, property (G1s) implies property (G1), although the converse implication is not true (see Theorem \ref{LinkageMethodsTheorem}).

On the other side, property (G2) ensures that
the dissimilarity degree between $\vec{\mathbb{X'}}$ and $\vec{\mathbb{Y}}$ does not change if we add to $\vec{\mathbb{X'}}$ another random variable that is comonotone with at least one element of $\vec{\mathbb{X'}}$.
Property (G2) is similar to the \emph{point proportion admissible} property considered for data points in \cite{FisNes71} that states that ``if after we duplicate one or more points any number of times and reapply the procedure the boundaries of the clusters are not changed at any stage.''. Here, in fact, we recall that two comonotonic r.v.'s are equal up to increasing transformations (see, e.g., \cite{DurSem16}).

Obviously, because of the symmetry of the dissimilarity function in \eqref{eq:simmetry}, properties (G1), (G1s) and (G2) can be also reformulated for the second argument of the involved dissimilarity functions.


\section{Extended dissimilarity functions: properties and examples}\label{sec:example}

In the following section, we provide various examples of dissimilarity functions and we study whether they satisfy some of the previously introduced properties.
Recall that, in view of property \eqref{local.diss.3}, any $(m_1,m_2)$-dissimilarity function aiming at quantifying the proximity degree of two random vectors $\vec{\mathbb{X}}$ and $\vec{\mathbb{Y}}$ of dimension $m_1$ and $m_2$, respectively, only depends on the $(m_1+m_2)$-dimensional copula $C$ of the random vector $(\vec{\mathbb{X}},\vec{\mathbb{Y}})$. Thus, in some cases, it could be also convenient to define the dissimilarity functions directly in terms of $C$.

Moreover, for the sake of a concise use of copulas and their margins,
for $ L \subseteq \{1,...,m\} $, we define the map
$ \boldsymbol{\eta}_{L}: \mathbb{I}^{m} \times \mathbb{I}^{m} \to \mathbb{I}^{m} $ given coordinatewise by
\begin{equation*}
  \big( \boldsymbol{\eta}_{L} (\uuu,\vvv) \big)_{\ell} :=
  \begin{cases}
    u_{\ell} & \ell \in \{1,...,m\} \backslash L     \\
    v_{\ell} & \ell \in L
  \end{cases}
\end{equation*}
and, for $ l \in \{1,...,m\} $, we put $ \boldsymbol{\eta}_{l} := \boldsymbol{\eta}_{\{l\}} $.
We denote by
$ \nul $ the vector with all entries equal to $0$,
by $ \ein $ the vector with all entries equal to $1$ and by $ \CC^k $ the collection of all $k$-dimensional copulas, $ 2 \leq k \leq m $.
For any subset $ L = \{l_{1},...,l_{|L|}\} \subseteq \{ 1, \dots, m\} $ with $ 2 \leq |L| \leq m $
such that $ l_{i} < l_{j} $ for all $ i,j \in \{1,...,|L|\} $ with $ i < j $, we further define by $T_{L}(C)$ the lower dimensional margin of the copula $C$ related to the indices of the components of $C$ belonging to $L$.

\begin{example}\label{ExampleMargins}$\,$
\begin{itemize}
\item
The identity $ T_{L} (M) = M $ holds for every $ L \subseteq \{1,...,m\} $ with $ 2 \leq |L| \leq m $.
\item The identity $ T_{L} (\Pi) = \Pi $ holds for every $ L \subseteq \{1,...,m\} $ with $ 2 \leq |L| \leq m $. Here, $\Pi$ is the independence copula given, for all $u_1,\dots,u_m$ in $\mathbb I$, by $\Pi(u_1,\dots,u_m)=\prod_{i=1}^m u_i$.
\item
For all $3\le m_1+m_2+m_3\le m$ and for
every random vector
$ (\vec{\mathbb{X'}}, \vec{\mathbb{X''}}, \vec{\mathbb{Y}})
  \in L^{0}(\mathbb{I}^{m_{1}}) \times L^{0}(\mathbb{I}^{m_{2}}) \times L^{0}(\mathbb{I}^{m_{3}}) $
with copula $ C_{(\vec{\mathbb{X'}}, \vec{\mathbb{X''}}, \vec{\mathbb{Y}})} \in \CC^{m_1 + m_2 + m_3} $,
the copulas $ C_{(\vec{\mathbb{X'}}, \vec{\mathbb{Y}})} \in \CC^{m_1 + m_3} $ and
$ C_{(\vec{\mathbb{X''}}, \vec{\mathbb{Y}})} \in \CC^{m_2 + m_3} $ satisfy
$ C_{(\vec{\mathbb{X'}}, \vec{\mathbb{Y}})}
  = T_{\{1,...,m_1+m_2+m_3\} \backslash \{m_1+1,...,m_1+m_2\}} \big( C_{(\vec{\mathbb{X}}, \vec{\mathbb{Y}})} \big) $
and \linebreak
$  C_{(\vec{\mathbb{X''}}, \vec{\mathbb{Y}})}
  = T_{\{m_1+1,...,m_1+m_2+m_3\}} \big( C_{(\vec{\mathbb{X}}, \vec{\mathbb{Y}})} \big). $
\end{itemize}
\end{example}{}

For every $k$ such that $2\le k\le m$, we further define the map $ [\cdot\,,\cdot]: \CC^k \times \CC^k \to \mathbb{R} $ introduced, e.g., in \cite{fuc2016a} and given by
$$
[C,D] := \int_{\mathbb{I}^{k}} C ({\bf u}) \; \mathrm{d} Q^D({\bf u})
$$
where $Q^D$ denotes the probability measure associated with the copula $D$.
The map $ [.\,,.] $ is linear with respect to convex combinations in both arguments and is therefore called a {\it biconvex form}.
Moreover, the map $ [\cdot\,,\cdot] $ satisfies $ [M,M] = 1/2 $ and $ [\Pi,\Pi] = 1/2^{k} $.

The following technical result will be needed in the following and it is reported here.

\begin{lemma}\label{AppendixLemma1}
Consider $ 2 \leq k \leq m $ and $ C \in \CC^{k} $ satisfying $ T_{\{i,j\}} (C)= M $
for some $ i, j \in \{1,...,k\} $ with $ i \neq j $.
\begin{itemize}
\item[(i)] Then $ Q^C \big[ \big\{ \uuu \in \mathbb{I}^k \, \big| \, u_i = u_j \big\} \big] = 1 $.
\item[(ii)] The identity
$ \int_{\mathbb{I}^k} f(\uuu) \; \mathrm{d} Q^C (\uuu)
	= \int_{\mathbb{I}^k} f (\boldsymbol{\eta}_{i} (\uuu, u_j \, \eee_i)) \; \mathrm{d} Q^C (\uuu) $
holds for every measurable function $ f: \mathbb{I}^k \to \R $.
\item[(iii)] The identity
$ C(u \, {\bf 1})
	= C (\boldsymbol{\eta}_{i} ( u \, {\bf 1}, {\bf 1} )) $
holds for every $ u \in \mathbb{I} $. 
\item[(iv)] Then
$ [C,C] = [T_{\{1,...,k\} \backslash\{i\}} (C), T_{\{1,...,k\} \backslash\{i\}} (C) ] $.
\item[(v)] The identity $ T_{i,l} (C)  = T_{j,l} (C) $
holds for every $ l \in \{1,...,k\} \backslash \{i,j\} $.
\end{itemize}
\end{lemma}

\begin{proof}
For $ p,q \in \{1,...,k\} $ with $ p \neq q $,
we define the projection
$ {\rm proj}_{\{p,q\}}: \mathbb{I}^k \to \mathbb{I}^2$,
$ {\rm proj}_{\{p,q\}} (\uuu) := (u_p, u_q)$.
Then
$ (Q^C)_{{\rm proj}_{\{i,j\}}} [[0,v_1] \times [0,v_2]]
	= C (\boldsymbol{\eta}_{\{i,j\}} ({\bf 1}, v_1 \, \eee_i + v_2 \, \eee_j))
	= (T_{\{i,j\}} (C)) (v_1,v_2)
	= M (v_1,v_2) $
for every $ \vvv \in \mathbb{I}^2 $ and hence
$ (Q^C)_{{\rm proj}_{\{i,j\}}}
  = Q^{T_{\{i,j\}} (C)}
	= Q^M $
which implies
$$
  Q^C [\{ \uuu \in \mathbb{I}^k \, | \, u_i < u_j \}]
	= (Q^C)_{{\rm proj}_{\{i,j\}}} [\{ \vvv \in \mathbb{I}^2 \, | \, v_1 < v_2 \}]
	= Q^M [\{ \vvv \in \mathbb{I}^2 \, | \, v_1 < v_2 \}]
	= 0
$$
Thus,
$ Q^C [\{\uuu \in \mathbb{I}^k \, | \, u_i = u_j\}] = 1 $
which proves (i) and, immediately, implies (ii).
Now, consider $ u \in \mathbb{I} $. Then, (ii) yields
\begin{eqnarray*}
  C(u \, {\bf 1})
	& = & \int_{\mathbb{I}^{k}} \chi_{[{\bf 0}, u \, {\bf 1}]} (\vvv) \; \mathrm{d} Q^C (\vvv)
	\\
	& = & \int_{\mathbb{I}^{k}} \chi_{[{\bf 0}, u \, {\bf 1}]} \big( \boldsymbol{\eta}_{i} (\vvv, v_j \, \eee_i) \big) \; \mathrm{d} Q^C (\vvv)
	\\
	& = & \int_{\mathbb{I}^{k}} \prod_{l=1, l \neq i}^k \chi_{[0,u]} (v_l) \; \mathrm{d} Q^C (\vvv)
	= C (\boldsymbol{\eta}_{i} (u \, {\bf 1}, {\bf 1}))
\end{eqnarray*}
where $\chi_B$ denotes the indicator function with respect to the set $B$. This proves (iii).
Moreover, (ii) together with \cite[Theorem 5.3.1]{fuc2015} yields
\begin{eqnarray*}
  [C,C]
	& = & \int_{\mathbb{I}^{k}} C (\uuu) \; \mathrm{d} Q^{C} (\uuu)
	= \int_{\mathbb{I}^{k}} C (\boldsymbol{\eta}_{i} (\uuu, u_j \, \eee_i)) \; \mathrm{d} Q^{C} (\uuu)
	\\
	& = & \int_{\mathbb{I}^{k}} \int_{\mathbb{I}^{k}}
	     \chi_{[{\bf 0}, \boldsymbol{\eta}_{i} (\uuu, u_j \, \eee_i)]} (\vvv) \; \mathrm{d} Q^{C} (\vvv) \mathrm{d} Q^{C} (\uuu)
	\\
	& = & \int_{\mathbb{I}^{k}} \int_{\mathbb{I}^{k}}
	     \chi_{[{\bf 0}, \boldsymbol{\eta}_{i} (\uuu, u_j \, \eee_i)]} (\boldsymbol{\eta}_{i} (\vvv, v_j \, \eee_i)) \; \mathrm{d} Q^{C} (\vvv) \mathrm{d} Q^{C} (\uuu)
	\\
	& = & \int_{\mathbb{I}^{k}} \int_{\mathbb{I}^{k}}
	      \prod_{l=1, l \neq i}^k \chi_{[0, u_l]} (v_l) \; \mathrm{d} Q^{C} (\vvv) \mathrm{d} Q^{C} (\uuu)
	\\
	& = & \int_{\mathbb{I}^{k}} \int_{\mathbb{I}^{k}}
	     \chi_{[{\bf 0}, \boldsymbol{\eta}_{i} (\uuu, 1 \, \eee_i)]} (\vvv) \; \mathrm{d} Q^{C} (\vvv) \mathrm{d} Q^{C} (\uuu)
	= \int_{\mathbb{I}^{k}} C (\boldsymbol{\eta}_{i} (\uuu, 1 \, \eee_i)) \; \mathrm{d} Q^{C} (\uuu)
	\\
	& = & [T_{\{1,...,k\} \backslash\{i\}} (C), T_{\{1,...,k\} \backslash\{i\}} (C)]
\end{eqnarray*}
This proves (iv).
Finally, consider $ l \in \{1,...,k\} \backslash \{i,j\} $.
Applying (ii) we obtain
\begin{eqnarray*}
  \big( T_{i,l} (C) \big) (v_1,v_2)
	& = & \int_{\mathbb{I}^2} \chi_{[0,v_1] \times [0,v_2]} (w_1,w_2) \; \mathrm{d} Q^{T_{i,l} (C)} (w_1,w_2)
	\\
	& = & \int_{\mathbb{I}^2} \chi_{[0,v_1] \times [0,v_2]} (w_1,w_2) \; \mathrm{d} \big( Q^C \big)_{{\rm proj}_{\{i,l\}}}  (w_1,w_2)
	\\
	& = & \int_{\mathbb{I}^k} \chi_{[0,v_1] \times [0,v_2]} ({\rm proj}_{\{i,l\}}(\uuu)) \; \mathrm{d} Q^C (\uuu)
	\\
	& = & \int_{\mathbb{I}^k} \chi_{[0,v_1] \times [0,v_2]} (u_i,u_l) \; \mathrm{d} Q^C (\uuu)
	\\
	& = & \int_{\mathbb{I}^k} \chi_{[0,v_1] \times [0,v_2]} (u_j,u_l) \; \mathrm{d} Q^C (\uuu)
	= \big( T_{j,l} (C) \big) (v_1,v_2)
\end{eqnarray*}
for every $ \vvv \in \mathbb{I}^2 $. This proves (v).
\end{proof}


\subsection{Extended dissimilarity functions based on linkage methods and a pairwise dissimilarity function}

First, we introduce dissimilarity functions that are defined in a similar way as in the classical hierarchical clustering algorithms, i.e. via single, average and complete linkage.

Consider a $(1,1)$-dissimilarity function $ d^{1,1} $ and $ m_1, m_2 \in \mathbb{N} $ with $ 2 \leq m_1 + m_2 \leq m $.
We define the maps
$ d^{m_1,m_2}_{\rm min}, d^{m_1,m_2}_{\rm ave}, d^{m_1,m_2}_{\rm max} \colon
  \Lina\times\Linb\to\mathbb{R}_+  $ by letting
\begin{eqnarray*}
  d^{m_1,m_2}_{\rm min} (\vec{\mathbb{X}},\vec{\mathbb{Y}})
	& := & \min \big\{ d^{1,1} (X,Y) \, \big| \, X \in \mathbb{X}, Y \in \mathbb{Y} \big\}
	\\
  d^{m_1,m_2}_{\rm ave} (\vec{\mathbb{X}},\vec{\mathbb{Y}})
	& := & \frac{1}{m_1 \, m_2} \sum_{X \in \mathbb{X}} \sum_{Y \in \mathbb{Y}} d^{1,1} (X,Y)
  \\
  d^{m_1,m_2}_{\rm max} (\vec{\mathbb{X}},\vec{\mathbb{Y}})
	& := & \max \big\{ d^{1,1} (X,Y) \, \big| \, X \in \mathbb{X}, Y \in \mathbb{Y} \big\}
\end{eqnarray*}
It is straightforward to show that,
for all $ 2 \leq m_1 + m_2 \leq m $,
$ d^{m_1,m_2}_{\rm min} $,
$ d^{m_1,m_2}_{\rm ave} $ and
$ d^{m_1,m_2}_{\rm max} $ are $(m_1,m_2)$-dissimilarity functions.
Thus, they can be extended as mappings from \linebreak
$\bigcup_{2 \leq m_1+m_2 \leq m} \Lina\times\Linb$ to $[0,+\infty[$
denoted, respectively, by $d_{\rm min}, d_{\rm ave}, d_{\rm max}$.
The mappings $d_{\rm min}$, $d_{\rm ave}$ and $d_{\rm max}$ are called, respectively, the \emph{single, average and complete extended dissimilarity functions induced by $d^{1,1}$}.

In the sequel, we focus on the extended dissimilarity function based on the following $(1,1)$-dissimilarity functions (see Section \ref{SubsectionExtDissMOD})
\begin{align}\label{eq:Blomqvist}
d^{1,1}_{\beta} (X,Y)
&:= \frac{1}{2} - C_{(X,Y)} \big( \tfrac{1}{2}, \tfrac{1}{2} \big)\\\label{eq:Footrule}
d^{1,1}_{\phi} (X,Y)
&:= \frac{1}{2} - \big[ C_{(X,Y)}, M \big]\\
\label{eq:Kendall}
d^{1,1}_{\tau} (X,Y)
&:= \frac{1}{2} - \big[ C_{(X,Y)}, C_{(X,Y)} \big]\\
\label{eq:Spearman}
d^{1,1}_{\rho} (X,Y)
& := \frac{1}{3} - \big[ C_{(X,Y)}, \Pi \big]
\end{align}
The function $d^{1,1}_{\beta}$ is related to the pairwise version of medial correlation coefficient (also known as Blomqvist's beta),
$d^{1,1}_{\phi}$ is related to the pairwise version of Spearman's footrule, and the functions
$d^{1,1}_{\tau}$ and $d^{1,1}_{\rho}$ are related to pairwise Kendall's tau and pairwise Spearman's rho.
In Section \ref{SubsectionExtDissMOD} we list some properties of these $(1,1)$-dissimilarity functions.

In the following we study whether single, average and complete extended dissimilarity functions satisfy some desirable properties;
having in mind that, in the bivariate case, lower orthant and concordance order coincide, the next result is straightforward.

\begin{theorem} \label{MoDLinkage}
Let $d_{\rm min}$, $d_{\rm ave}$ and $d_{\rm max}$ be the extended dissimilarity functions induced by $d^{1,1}$. Then:
\begin{itemize}
\item[(i)] $ d_{\rm min} $, $ d_{\rm ave} $ and $ d_{\rm max} $ satisfy (L1) and (L1c) whenever $ d^{1,1} $ is decreasingly monotone with respect to lower orthant order;

\item[(ii)] $ d_{\rm min} $, $ d_{\rm ave} $  and $ d_{\rm max} $ satisfy ($L2$) whenever $d^{1,1}(X,Y)=d^{1,1}(1-X,1-Y)$ for all $X,Y\in L^{0}(\mathbb{I})$;

\item[(iii)] $ d_{\rm min} $, $ d_{\rm ave} $ and $ d_{\rm max} $ satisfy ($L3$) whenever $ d^{1,1}$ is continuous with respect to weak convergence.
\end{itemize}
\end{theorem}

In the following theorem we show that the single,
the average and the complete extended dissimilarity functions satisfy some of the global properties introduced above.

\begin{theorem}\label{LinkageMethodsTheorem}
Let $d_{\rm min}$, $d_{\rm ave}$ and $d_{\rm max}$ be the extended dissimilarity functions induced by $d^{1,1}$. Then:
\begin{itemize}
\item[(i)] $d_{\rm min}$  satisfies (G1) and (G2),
but fails to satisfy (G1s);
\item[(ii)] $d_{\rm ave}$ satisfies (G1);
\item[(iii)] $d_{\rm max}$ satisfies (G1) and (G2).
\end{itemize}
\end{theorem}

\begin{proof}
We first prove (G1). To this end, consider
$ 3 \le m_{1}+m_{2}+m_3 \le m $,
the random vector
$ (\vec{\mathbb{X'}}, \vec{\mathbb{X''}}, \vec{\mathbb{Y}})
  \in L^{0}(\mathbb{I}^{m_{1}}) \times L^{0}(\mathbb{I}^{m_{2}}) \times L^{0}(\mathbb{I}^{m_{3}})$
satisfying
$ d^{m_{1},m_{2}} (\vec{\mathbb{X'}},\vec{\mathbb{X''}}) \leq \linebreak
\min\big\{ d^{m_{1},m_3} (\vec{\mathbb{X'}},\vec{\mathbb{Y}}), d^{m_{2},m_3} (\vec{\mathbb{X''}},\vec{\mathbb{Y}}) \big\} $
such that $ \mathbb{X'} $, $\mathbb{X''} $ and $\mathbb{Y} $ are pairwise disjoint
and put $ \mathbb{X} := \mathbb{X'} \cup \mathbb{X''} $.
Then
\begin{eqnarray*}
  \lefteqn{\min\big\{ d_{\rm min}^{m_{1},m_3} (\vec{\mathbb{X'}},\vec{\mathbb{Y}}),
											d_{\rm min}^{m_{2},m_3} (\vec{\mathbb{X''}},\vec{\mathbb{Y}}) \big\}}
	\\
	& = & \min\Big\{ \min \big\{ d^{1,1} (X,Y) \, \big| \, X \in \mathbb{X'}, Y \in \mathbb{Y} \big\},									\min \big\{ d^{1,1} (X,Y) \, \big| \, X \in \mathbb{X''}, Y \in \mathbb{Y} \big\} \Big\}
	\\
	& = & \min \big\{ d^{1,1} (X,Y) \, \big| \, X \in \mathbb{X}, Y \in \mathbb{Y} \big\}
	\\
	& = & d^{m_1+m_2,m_3}_{\rm min} (\vec{\mathbb{X}},\vec{\mathbb{Y}}).
\end{eqnarray*}
Thus, $d_{\rm min}$ satifies (G1), but cannot satisfy (G1s).
Moreover,
\begin{eqnarray*}
  \lefteqn{\min\big\{ d_{\rm ave}^{m_{1},m_3} (\vec{\mathbb{X'}},\vec{\mathbb{Y}}),
											d_{\rm ave}^{m_{2},m_3} (\vec{\mathbb{X''}},\vec{\mathbb{Y}}) \big\}}
	\\
	&   =  & \min\left\{ \frac{1}{m_1 \, m_3} \sum_{X \in \mathbb{X'}} \sum_{Y \in \mathbb{Y}} d^{1,1} (X,Y), 										 \frac{1}{m_2 \, m_3} \sum_{X \in \mathbb{X''}} \sum_{Y \in \mathbb{Y}} d^{1,1} (X,Y) \right\}
  \\
	& \leq & \frac{m_1}{m_1+m_2} \; \frac{1}{m_1 \, m_3} \sum_{X \in \mathbb{X'}} \sum_{Y \in \mathbb{Y}} d^{1,1} (X,Y) + 		 \frac{m_2}{m_1+m_2} \; \frac{1}{m_2 \, m_3} \sum_{X \in \mathbb{X''}} \sum_{Y \in \mathbb{Y}} d^{1,1} (X,Y)
	\\
	&   =  & \frac{1}{(m_1+m_2) \, m_3} \sum_{X \in \mathbb{X}} \sum_{Y \in \mathbb{Y}} d^{1,1} (X,Y)
	\\*
	&   =  & d^{m_1+m_2,m_3}_{\rm ave} (\vec{\mathbb{X}},\vec{\mathbb{Y}}),
  \\
  \lefteqn{\min\big\{ d_{\rm max}^{m_{1},m_3} (\vec{\mathbb{X'}},\vec{\mathbb{Y}}),
											d_{\rm max}^{m_{2},m_3} (\vec{\mathbb{X''}},\vec{\mathbb{Y}}) \big\}}
	\\*
	&   =  & \min\Big\{ \max \big\{ d^{1,1} (X,Y) \, \big| \, X \in \mathbb{X'}, Y \in \mathbb{Y} \big\},									\max \big\{ d^{1,1} (X,Y) \, \big| \, X \in \mathbb{X''}, Y \in \mathbb{Y} \big\} \Big\}
	\\
	& \leq & \max\Big\{ \max \big\{ d^{1,1} (X,Y) \, \big| \, X \in \mathbb{X'}, Y \in \mathbb{Y} \big\},									\max \big\{ d^{1,1} (X,Y) \, \big| \, X \in \mathbb{X''}, Y \in \mathbb{Y} \big\} \Big\}
	\\
	&   =  & \max \big\{ d^{1,1} (X,Y) \, \big| \, X \in \mathbb{X}, Y \in \mathbb{Y} \big\}
	\\*
	&   =  & d^{m_1+m_2,m_3}_{\rm max} (\vec{\mathbb{X}},\vec{\mathbb{Y}}).
\end{eqnarray*}
Thus, the average and complete extended dissimilarity functions satisfy  (G1).

Now, we prove property (G2).
To this end, consider $3\le m_1+m_2\le m$ with $ 2 \le m_1 $,
$\vec{\mathbb{X'}} \in L^{0}(\mathbb{I}^{m_1-1}) $ and
$ (\vec{\mathbb{X}},\vec{\mathbb{Y}}) \in \Lina \times \Linb $ such that
$ \mathbb{X'} \cup \{X''\} = \mathbb{X} $, where $ X'' \in \mathbb{X} $ is comonotonic with some element
$ X' \in \mathbb{X'}$.
Then, by Lemma \ref{AppendixLemma1} and property \eqref{local.diss.3}, the identity
 $d^{1,1} (X',Y)=d^{1,1}(X'',Y)$ holds for every $Y \in \mathbb{Y}$, and we obtain
\begin{eqnarray*}
  d^{m_1,m_2}_{\rm min} (\vec{\mathbb{X}},\vec{\mathbb{Y}})
	& = & \min \big\{ d^{1,1} (X,Y) \, \big| \, X \in \mathbb{X}, Y \in \mathbb{Y} \big\}
	\\*
	& = & \min \big\{ d^{1,1} (X,Y) \, \big| \, X \in \mathbb{X'}, Y \in \mathbb{Y} \big\}
	\\*
	& = & d^{m_1-1,m_2}_{\rm min} (\vec{\mathbb{X'}},\vec{\mathbb{Y}})
  \\
	d^{m_1,m_2}_{\rm max} (\vec{\mathbb{X}},\vec{\mathbb{Y}})
	& = & \max \big\{ d^{1,1} (X,Y) \, \big| \, X \in \mathbb{X}, Y \in \mathbb{Y} \big\}
	\\
	& = & \max \big\{ d^{1,1} (X,Y) \, \big| \, X \in \mathbb{X'}, Y \in \mathbb{Y} \big\}
	\\*
	& = & d^{m_1-1,m_2}_{\rm max} (\vec{\mathbb{X'}},\vec{\mathbb{Y}})
\end{eqnarray*}
Thus, the single and complete extended dissimilarity functions satisfy (G2).
\end{proof}

We now present some sufficient condition on $d^{1,1}$ such that both the average and complete extended dissimilarity functions satisfy (G1s).

\begin{corollary}\label{cor:HierarchicalClusterG2}
Assume that $d^{1,1}$ is strictly monotonically decreasing with respect to the lower orthant order,
i.e. $(X,Y) \prec_{lo} (X',Y')$ in the lower orthant order implies
$ d^{1,1} (X',Y') < d^{1,1} (X,Y) $.
Then $d_{\rm ave}$ and $d_{\rm max}$ satisfy (G1s).
\end{corollary}

\begin{proof}
Consider $m\geq 3 $, $ m_1 = m_2 = m_3 = 1 $ and the $m$--dimensional  copula $ C$ given by
\begin{equation}\label{eq:PertubatedPi}
  C(\uuu)
	:= \Pi(\uuu)
	    - \frac{1}{3} \; \Big( (1-u_1)(1-u_2) + (1-u_1)(1-u_3) + (1-u_2)(1-u_3) \Big)
			  \prod_{i=1}^{3} u_i^{4-i} \prod_{i=4}^{m} u_i.
\end{equation}
(To check that this function is actually a copula it is enough to compute its density).
Then, for every random vector
$ (X',X'',Y) \in L^{0}(\mathbb{I}) \times L^{0}(\mathbb{I}) \times L^{0}(\mathbb{I}) $ having copula
$ T_{\{1,2,3\}} (C) $, we have
$$
 (X'',Y)
	\prec_{lo} (X',Y)
	\prec_{lo} (X',X'')
$$
and hence
$d^{1,1} (X',X'') < \min\big\{ d^{1,1} (X',Y), d^{1,1} (X'',Y) \big\}$ as well as
\begin{center}
\begin{tabular}{lclcl}
  $\min\big\{ d^{1,1} (X',Y), d^{1,1} (X'',Y) \big\}$
	& $<$ & $\frac{1}{2} \, d^{1,1} (X',Y) + \frac{1}{2} \, d^{1,1} (X'',Y)$
	& $=$ & $d^{1+1,1}_{\rm ave} (\vec{\mathbb{X}},Y)$
	\\
  $\min\big\{ d^{1,1} (X',Y), d^{1,1} (X'',Y) \big\}$
	& $<$ & $\max\big\{ d^{1,1} (X',Y), d^{1,1} (X'',Y) \big\}$
	& $=$ & $d^{1+1,1}_{\rm max} (\vec{\mathbb{X}},Y)$
\end{tabular}
\end{center}
where $ \mathbb{X} = X' \cup X'' $.
Therefore, $d_{\rm ave}$ and $d_{\rm max}$ satisfy (G1s).
\end{proof}

\begin{remark}
Notice that the single, average and complete extended dissimilarity functions induced by $d^{1,1} := f \circ \rho$, where $\rho$ is the pairwise Spearman's correlation and $f$ is a strictly decreasing function, are strictly monotonically decreasing with respect to the lower orthant order (see, e.g., \cite{ahnfuchs2018}).
\end{remark}

The following example shows that the condition stated in Corollary \ref{cor:HierarchicalClusterG2} is sufficient, but not necessary.

\begin{example}\label{HierarchicalClusterCounterEx}
Consider the map $d^{1,1}_{\beta}$ given by \eqref{eq:Blomqvist},
$m \geq 3 $, $ m_1 = m_2 = m_3 = 1 $ and the copula $C$ given by \eqref{eq:PertubatedPi}.
Further, note that $d^{1,1}_{\beta}$ fails to be strictly monotonically decreasing with respect to the lower orthant order: to verify this, it is enough to consider two copulas with the same value in the point $(0.5,0.5)$, like ordinal sums of two copulas with respect to the partition $([0,0.5],[0.5,1])$. Then, for every random vector
$ (X',X'',Y) \in L^{0}(\mathbb{I}) \times L^{0}(\mathbb{I}) \times L^{0}(\mathbb{I}) $ having copula
$ T_{\{1,2,3\}} (C) $, we have
\begin{center}
\begin{tabular}{lclcl}
  $d^{1,1}_{\beta} (X',X'')$
	& $\leq$ & $\min\big\{ d^{1,1}_{\beta} (X',Y), d^{1,1}_{\beta} (X'',Y) \big\}$
	&   $<$  & $d^{1+1,1}_{\rm ave} (\vec{\mathbb{X}},Y)$	
	\\
	$d^{1,1}_{\beta} (X',X'')$
	& $\leq$ & $\min\big\{ d^{1,1}_{\beta} (X',Y), d^{1,1}_{\beta} (X'',Y) \big\}$
	&   $<$  & $d^{1+1,1}_{\rm max} (\vec{\mathbb{X}},Y)$
\end{tabular}
\end{center}
where $ \mathbb{X} = X' \cup X'' $.
Indeed, we have
\begin{eqnarray*}
  d^{1,1}_{\beta} (X',X'')
	&   =  & \frac{1}{4} + \frac{1}{3} \; \frac{1}{2^7}
	\\
	\min\big\{ d^{1,1}_{\beta} (X',Y), d^{1,1}_{\beta} (X'',Y) \big\}
	&   =  & \frac{1}{4} + \frac{1}{3} \; \frac{2}{2^7}
  \\
	d^{1+1,1}_{\rm ave} (\vec{\mathbb{X}},Y)	
  &   =  & \frac{1}{4} + \frac{1}{3} \; \frac{3}{2^7}
  \\
	d^{1+1,1}_{\rm max} (\vec{\mathbb{X}},Y)	
  &   =  & \frac{1}{4} + \frac{1}{3} \; \frac{4}{2^7}
\end{eqnarray*}
Thus, the average and complete extended dissimilarity functions induced by $d^{1,1}_{\beta}$ satisfy (G1s).
By applying the above copula,
it is straightforward to check that also the average and complete extended dissimilarity functions induced by
$d^{1,1}_{\phi}$ and $d^{1,1}_{\tau}$ given by (\ref{eq:Footrule}) and (\ref{eq:Kendall}), satisfy (G1s). We notice that such dissimilarity functions both fail to be strictly monotonically decreasing with respect to the lower orthant order.
\end{example}

The next example shows that the average extended dissimilarity function may fail to satisfy (G2) for specific choices of $d^{1,1}$.

\begin{example}\label{HierarchicalClusterCounterEx2}
Consider $d^{1,1}_{\beta}$ given by \eqref{eq:Blomqvist}. Further, consider $m \geq 4 $, $ m_1 = 3 $, $ m_2 = 1 $ and the $m$--dimensional copula $C$ given by
$$
  C(\uuu)
	:= \min\{u_1, u_2\} \, \left( \prod_{i=3}^m u_i - \frac{1}{2} \; (1-u_3)(1-u_4) \prod_{i=3}^m u_i \right)
$$
which is the copula of a random vector with two independent sub-vectors (see e.g. \cite{DurSem16}).
Then, for every random vector
$ \vec{\mathbb{X}} = (X_1,X_2,X_3)^{\prime} \in L^{0}(\mathbb{I}^{3}) $
and every r.v.~$ Y \in L^{0}(\mathbb{I}) $
such that $ (\vec{\mathbb{X}},Y) $ has copula $ T_{\{1,2,3,4\}} (C) $ and hence $X_1$ and $X_2$ are comonotonic,
the average extended dissimilarity function induced by $d^{1,1}_{\beta}$ satisfies
$$
  d^{3,1}_{\rm ave} (\vec{\mathbb{X}},Y)
	\neq d^{2,1}_{\rm ave} ((X_1,X_3),Y)
$$
Indeed, we obtain
$
  d^{3,1}_{\rm ave} (\vec{\mathbb{X}},Y)
	  =  \frac{1}{4} + \frac{1}{3} \; \frac{1}{2^5}
	\neq \frac{1}{4} + \frac{1}{2} \; \frac{1}{2^5}
	  =  d^{2,1}_{\rm ave} ((X_1,X_3),Y).
$
Thus, the average extended dissimilarity function based on $d^{1,1}_{\beta}$ fails to satisfy (G2).
By applying the above copula,
it is straightforward to check that also the average extended dissimilarity functions induced by
$d^{1,1}_{\phi}$, $d^{1,1}_{\tau}$ and $d^{1,1}_{\rho}$ given, respectively, by \eqref{eq:Footrule}, \eqref{eq:Kendall} and \eqref{eq:Spearman} fail to satisfy (G2).
\end{example}

We conclude by noticing that the extended dissimilarity functions based on single, average and complete linkage share the same structural drawback: They take into account solely information about the pairwise dependence structure. Therefore, for each of these extended dissimilarity functions, the value of an $(m_1,m_2)$-dissimilarity function of a random vector depends on its bivariate margins only. The next result is hence evident.

\begin{corollary}
Consider a $(1,1)$-dissimilarity function $ d^{1,1} $,
some constant $ c\ge 0$,
$2\le m_1+m_2\le m$ and let
$ (\vec{\mathbb{X}}, \vec{\mathbb{Y}}) \in \Lina \times \Linb $ be a random vector satisfying
$ d^{1,1} (X,Y) = c $
for every $ X \in \mathbb{X} $ and every $ Y \in \mathbb{Y} $.
Then
$ d^{m_1,m_2}_{\rm min} (\vec{\mathbb{X}},\vec{\mathbb{Y}})
	= d^{m_1,m_2}_{\rm ave} (\vec{\mathbb{X}},\vec{\mathbb{Y}})
	= d^{m_1,m_2}_{\rm max} (\vec{\mathbb{X}},\vec{\mathbb{Y}})
	= c $.
\end{corollary}

\begin{example}\label{MoDIndependenceCounterExample}
Consider $2\le m_1+m_2 \leq m$ with $m \geq 3$ and the copula
$C: \mathbb{I}^{m} \to \mathbb{I}$ given by
$$
  C(\uuu)
	:= \Pi(\uuu) + \prod_{i=1}^{m} u_{i} (1-u_{i})
$$
Then
$ C \neq \Pi $ and since $m \geq 3$ we have $ T_{L} (C) = \Pi $ for every $ L \subseteq \{1,...,m\} $ with $|L| = 2$ (note that the second term on the right hand side vanishes when putting $u_j = 1$ for some $j \in \{1,\dots,m\} \backslash L$), and the identities
\begin{eqnarray*}
  d^{m_1,m_2}_{\rm min} (\vec{\mathbb{X}},\vec{\mathbb{Y}})
	& = & d^{m_1,m_2}_{\rm min} (\vec{\mathbb{X'}},\vec{\mathbb{Y'}})
	\\
	d^{m_1,m_2}_{\rm ave} (\vec{\mathbb{X}},\vec{\mathbb{Y}})
	& = & d^{m_1,m_2}_{\rm ave} (\vec{\mathbb{X'}},\vec{\mathbb{Y'}})
	\\
	d^{m_1,m_2}_{\rm max} (\vec{\mathbb{X}},\vec{\mathbb{Y}})
	& = & d^{m_1,m_2}_{\rm max} (\vec{\mathbb{X'}},\vec{\mathbb{Y'}})
\end{eqnarray*}
hold for every random vector
$ (\vec{\mathbb{X}}, \vec{\mathbb{Y}}) \in \Lina \times \Linb $ with copula $C$
and every random vector
$ (\vec{\mathbb{X'}}, \vec{\mathbb{Y'}}) \in \Lina \times \Linb $ with copula $\Pi$.
Thus, neither the single nor
the average nor the complete extended dissimilarity function distinguishes between pairwise independence and global independence.
\end{example}


\subsection{Extended dissimilarity functions based on measures of  multivariate association}
\label{SubsectionExtDissMOD}

In this section we study extended dissimilarity functions which are derived from various measures of multivariate association (see, e.g., \cite{Nel06,Sch_et_al10CTA}).
Contrarily to the dissimilarity functions based on linkage methods, here we rely on global measures of association which do not only depend on the pairwise association. Thus, in principle, the derived dissimilarity functions could be able to detect high-dimensional features that are not apparent with the latter methods.
To this end, for $ m_1, m_2 \in \mathbb{N} $ with $ 2 \leq m_1 + m_2 \leq m $,
we define the maps
$d^{m_1,m_2}_{\beta}$,
$d^{m_1,m_2}_{\phi}$,
$d^{m_1,m_2}_{\tau}$,
$d^{m_1,m_2}_{\rho}$
from $\Lina\times\Linb$ to $[0,+\infty[$ by letting
\begin{center}
\begin{tabular}{lclcl}
  $d^{m_1,m_2}_{\beta} (\vec{\mathbb{X}},\vec{\mathbb{Y}})$
	& $:=$ & $\frac{1}{2} - C_{(\vec{\mathbb{X}},\vec{\mathbb{Y}})} \big( \tfrac{\bf 1}{\bf 2} \big)$
	& $ =$ & $M \big( \tfrac{\bf 1}{\bf 2} \big)  - C_{(\vec{\mathbb{X}},\vec{\mathbb{Y}})} \big( \tfrac{\bf 1}{\bf 2} \big)$
	\\
	$d^{m_1,m_2}_{\phi} (\vec{\mathbb{X}},\vec{\mathbb{Y}})$
	& $:=$ & $\frac{1}{2} - \big[ C_{(\vec{\mathbb{X}},\vec{\mathbb{Y}})}, M \big]$
	& $ =$ & $\int_{\mathbb{I}} (M(u \,{\bf 1}) - C_{(\vec{\mathbb{X}},\vec{\mathbb{Y}})} (u \,{\bf 1})) \; \mathrm{d} \leb(u)$
	\\
	$d^{m_1,m_2}_{\tau} (\vec{\mathbb{X}},\vec{\mathbb{Y}})$
	& $:=$ & $\frac{1}{2} - \big[ C_{(\vec{\mathbb{X}},\vec{\mathbb{Y}})}, C_{(\vec{\mathbb{X}},\vec{\mathbb{Y}})} \big]$
	& $ =$ & $\big[ M, M \big] - \big[ C_{(\vec{\mathbb{X}},\vec{\mathbb{Y}})}, C_{(\vec{\mathbb{X}},\vec{\mathbb{Y}})} \big]$
	\\
	$d^{m_1,m_2}_{\rho} (\vec{\mathbb{X}},\vec{\mathbb{Y}})$
	& $:=$ & $\frac{1}{m_1+m_2+1} - \big[ C_{(\vec{\mathbb{X}},\vec{\mathbb{Y}})}, \Pi \big]$
	& $ =$ & $\int_{\mathbb{I}^{m_1+m_2}} (M(\uuu) - C_{(\vec{\mathbb{X}},\vec{\mathbb{Y}})} (\uuu)) \; \mathrm{d} \leb^{m_1+m_2}(\uuu)$
\end{tabular}
\end{center}
The function $d^{m_1,m_2}_{\beta}$ is related to the multivariate version of medial correlation coefficient (also known as Blomqvist's beta) that was introduced by \cite{Nel03} (see also \cite{Ube05}), whose $n$-dimensional version is given by $\big(2^n C\big( \tfrac{\bf 1}{\bf 2} \big)-1\big)/(2^{n-1}-1)$. The function $d^{m_1,m_2}_{\phi}$ is related to the multivariate version of Spearman's footrule considered in \cite{Ube05}. The functions $d^{m_1,m_2}_{\tau}$ and $d^{m_1,m_2}_{\rho}$ are related to some multivariate versions of Kendall's tau and Spearman's rho (see, for instance, \cite{Joe15,Sch_et_al10CTA,Tay16}).

\begin{theorem}\label{thm:global_d}
For all $ 2 \leq m_1 + m_2 \leq m $,
$ d^{m_1,m_2}_{\beta} $,
$ d^{m_1,m_2}_{\phi} $,
$ d^{m_1,m_2}_{\tau} $ and
$ d^{m_1,m_2}_{\rho} $
are $(m_1,m_2)$-dissimilarity functions,
and thus, the maps
$$
  d_{\beta},
	d_{\phi},
	d_{\tau},
	d_{\rho}:
	\bigcup_{2 \leq m_1+m_2 \leq m} \Lina\times\Linb\to[0,+\infty[
$$
with
$ d_{\beta}|_{\Lina\times\Linb} := d^{m_1,m_2}_{\beta} $,
$ d_{\phi}|_{\Lina\times\Linb} := d^{m_1,m_2}_{\phi} $,
$ d_{\tau}|_{\Lina\times\Linb} := d^{m_1,m_2}_{\tau} $ and
$ d_{\rho}|_{\Lina\times\Linb} := d^{m_1,m_2}_{\rho} $
are extended dissimilarity functions.
Moreover,
\begin{itemize}
\item[(i)] $ d_{\beta} $ satisfies (L1), (L1c), (L3), (G1), (G1s) and (G2).

\item[(ii)] $ d_{\phi} $ satisfies (L1), (L1c), (L3), (G1), (G1s) and (G2).

\item[(iii)] $ d_{\tau} $ satisfies (L1c), (L2), (L3),
(G1), (G1s) and (G2).

\item[(iv)] $ d_{\rho} $ satisfies (L1), (L1c), (L3), but fails to satisfy (G1), (G1s) and (G2).
\end{itemize}
\end{theorem}

\begin{proof}
We first prove the local properties and then, step by step, all the global properties.
Since
$C^{\xi} \big( \tfrac{\bf 1}{\bf 2} \big) = C \big( \tfrac{\bf 1}{\bf 2} \big)$,
$ [C^{\xi},M] = [C,M] $,
$ [C^{\xi},C^{\xi}] = [C,C] $ and
$ [C^{\xi},\Pi] = [C,\Pi] $
for every $ C \in \mathcal{C}^k $, every permutation $\xi$ of $\mathbb{I}^k$, where $C^\xi$ is the copula obtained from $C$ by permuting its arguments, and for every $ 2 \leq k \leq m $
(see \cite[Theorem 5.2]{fuc2016a}),
it follows that
$d_{\beta}$,
$d_{\phi}$,
$d_{\tau}$ and
$d_{\rho}$ are extended dissimilarity functions.
It is evident that $d_{\beta}$ satisfies (L1), (L1c) and (L3),
and it is immediate from \cite[Theorems 3.6, 4.3, 4.4 and 5.2]{fuc2016a} that
$d_{\phi}$ and
$d_{\rho}$ satisfy (L1), (L1c) and (L3) and that $d_{\tau}$ satisfies (L1c), (L2) and (L3).

Now, consider
$ 3 \le m_{1}+m_{2}+m_3 \le m $,
the random vector
$ (\vec{\mathbb{X'}}, \vec{\mathbb{X''}}, \vec{\mathbb{Y}})
  \in L^{0}(\mathbb{I}^{m_{1}}) \times L^{0}(\mathbb{I}^{m_{2}}) \times L^{0}(\mathbb{I}^{m_{3}}) $
such that $ \mathbb{X'} $, $\mathbb{X''} $ and $\mathbb{Y} $ are pairwise disjoint and \eqref{eq:reducibility1} holds, and
put $ \mathbb{X} := \mathbb{X'} \cup \mathbb{X''} $.
Then
\begin{eqnarray*}
  C_{(\vec{\mathbb{X}},\vec{\mathbb{Y}})} \big( \tfrac{\bf 1}{\bf 2} \big)
  & \leq & C_{(\vec{\mathbb{X}},\vec{\mathbb{Y}})} \big( \boldsymbol{\eta}_{\{m_1+1,...,m_1+m_2\}} \big( \tfrac{\bf 1}{\bf 2}, {\bf 1} \big) \big)
	\\*
	& 	=  & \big( T_{\{1,...,m_1+m_2+m_3\} \backslash \{m_1+1,...,m_1+m_2\}} \big( C_{(\vec{\mathbb{X}}, \vec{\mathbb{Y}})} \big) \big) \big( \tfrac{\bf 1}{\bf 2} \big)
	= C_{(\vec{\mathbb{X'}},\vec{\mathbb{Y}})} \big( \tfrac{\bf 1}{\bf 2} \big)
  \\
  C_{(\vec{\mathbb{X}},\vec{\mathbb{Y}})} \big( \tfrac{\bf 1}{\bf 2} \big)
  & \leq & C_{(\vec{\mathbb{X}},\vec{\mathbb{Y}})} \big( \boldsymbol{\eta}_{\{1,...,m_1\}} \big( \tfrac{\bf 1}{\bf 2}, {\bf 1} \big) \big)
	\\
	&   =  & \big( T_{\{m_1+1,...,m_1+m_2+m_3\}} \big( C_{(\vec{\mathbb{X}}, \vec{\mathbb{Y}})} \big) \big) \big( \tfrac{\bf 1}{\bf 2} \big)
	= C_{(\vec{\mathbb{X''}},\vec{\mathbb{Y}})} \big( \tfrac{\bf 1}{\bf 2} \big)
\end{eqnarray*}
and, by \cite[Theorem 5.3.1]{fuc2015} and Example \ref{ExampleMargins}, we obtain
\begin{eqnarray*}
  \big[ C_{(\vec{\mathbb{X}},\vec{\mathbb{Y}})}, M \big]
	& \leq & \int_{\mathbb{I}^{m_1+m_2+m_3}} C_{(\vec{\mathbb{X}},\vec{\mathbb{Y}})} \big( \boldsymbol{\eta}_{\{m_1+1,...,m_1+m_2\}} (\uuu, {\bf 1}) \big) \; \mathrm{d} Q^M (\uuu)
  \\
	&   =  & \big[ T_{\{1,...,m_1+m_2+m_3\} \backslash \{m_1+1,...,m_1+m_2\}} \big( C_{(\vec{\mathbb{X}}, \vec{\mathbb{Y}})} \big), T_{\{1,...,m_1+m_2+m_3\} \backslash \{m_1+1,...,m_1+m_2\}} (M) \big]
	\\
	&   =  & \big[ C_{(\vec{\mathbb{X'}},\vec{\mathbb{Y}})}, M \big]
	\\
	\big[ C_{(\vec{\mathbb{X}},\vec{\mathbb{Y}})}, M \big]
	& \leq & \int_{\mathbb{I}^{m_1+m_2+m_3}} C_{(\vec{\mathbb{X}},\vec{\mathbb{Y}})} \big( \boldsymbol{\eta}_{\{1,...,m_1\}} (\uuu, {\bf 1}) \big) \; \mathrm{d} Q^M (\uuu)
  \\*
	&   =  & \big[ T_{\{m_1+1,...,m_1+m_2+m_3\}} \big( C_{(\vec{\mathbb{X}}, \vec{\mathbb{Y}})} \big),
								 T_{\{m_1+1,...,m_1+m_2+m_3\}} (M) \big]
	\\*
	&   =  & \big[ C_{(\vec{\mathbb{X''}},\vec{\mathbb{Y}})}, M \big]
\end{eqnarray*}
as well as
\begin{eqnarray*}
  \big[ C_{(\vec{\mathbb{X}},\vec{\mathbb{Y}})}, C_{(\vec{\mathbb{X}},\vec{\mathbb{Y}})} \big]
	& \leq & \int_{\mathbb{I}^{m_1+m_2+m_3}} C_{(\vec{\mathbb{X}},\vec{\mathbb{Y}})} \big( \boldsymbol{\eta}_{\{m_1+1,...,m_1+m_2\}} (\uuu, {\bf 1}) \big) \; \mathrm{d} Q^{C_{(\vec{\mathbb{X}},\vec{\mathbb{Y}})}} (\uuu)
  \\*
	&   =  & \big[ T_{\{1,...,m_1+m_2+m_3\} \backslash \{m_1+1,...,m_1+m_2\}} \big( C_{(\vec{\mathbb{X}}, \vec{\mathbb{Y}})} \big), T_{\{1,...,m_1+m_2+m_3\} \backslash \{m_1+1,...,m_1+m_2\}} (C_{(\vec{\mathbb{X}},\vec{\mathbb{Y}})}) \big]
	\\*
	&   =  & \big[ C_{(\vec{\mathbb{X'}},\vec{\mathbb{Y}})}, C_{(\vec{\mathbb{X'}},\vec{\mathbb{Y}})} \big]
  \\*
	\big[ C_{(\vec{\mathbb{X}},\vec{\mathbb{Y}})}, C_{(\vec{\mathbb{X}},\vec{\mathbb{Y}})} \big]
	& \leq & \int_{\mathbb{I}^{m_1+m_2+m_3}} C_{(\vec{\mathbb{X}},\vec{\mathbb{Y}})} \big( \boldsymbol{\eta}_{\{1,...,m_1\}} (\uuu, {\bf 1}) \big) \; \mathrm{d} Q^{C_{(\vec{\mathbb{X}},\vec{\mathbb{Y}})}} (\uuu)
  \\*
	&   =  & \big[ T_{\{m_1+1,...,m_1+m_2+m_3\}} \big( C_{(\vec{\mathbb{X}}, \vec{\mathbb{Y}})} \big),
								 T_{\{m_1+1,...,m_1+m_2+m_3\}} (C_{(\vec{\mathbb{X}},\vec{\mathbb{Y}})}) \big]
	\\*
	&   =  & \big[ C_{(\vec{\mathbb{X''}},\vec{\mathbb{Y}})}, C_{(\vec{\mathbb{X''}},\vec{\mathbb{Y}})} \big]
\end{eqnarray*}
Thus,
\begin{eqnarray*}
  \min\big\{ d^{m_{1},m_3}_{\beta} (\vec{\mathbb{X'}},\vec{\mathbb{Y}}),
						 d^{m_{2},m_3}_{\beta} (\vec{\mathbb{X''}},\vec{\mathbb{Y}}) \big\}
  & \leq & d^{m_{1}+m_{2},m_3}_{\beta} (\vec{\mathbb{X}},\vec{\mathbb{Y}})
	\\*
	\min\big\{ d^{m_{1},m_3}_{\phi} (\vec{\mathbb{X'}},\vec{\mathbb{Y}}),
						 d^{m_{2},m_3}_{\phi} (\vec{\mathbb{X''}},\vec{\mathbb{Y}}) \big\}
  & \leq & d^{m_{1}+m_{2},m_3}_{\phi} (\vec{\mathbb{X}},\vec{\mathbb{Y}})
	\\*
	\min\big\{ d^{m_{1},m_3}_{\tau} (\vec{\mathbb{X'}},\vec{\mathbb{Y}}),
						 d^{m_{2},m_3}_{\tau} (\vec{\mathbb{X''}},\vec{\mathbb{Y}}) \big\}
  & \leq & d^{m_{1}+m_{2},m_3}_{\tau} (\vec{\mathbb{X}},\vec{\mathbb{Y}})
\end{eqnarray*}
which implies that $ d_{\beta} $, $ d_{\phi} $ and $ d_{\tau} $ satisfy (G1).
In Examples \ref{HierarchicalMOCCounterEx} and \ref{HierarchicalMOCCounterExRho1} we show that the extended dissimilarity functions $ d_{\beta} $, $ d_{\phi} $ and $ d_{\tau} $ satisfy also (G1s),
and that $ d_{\rho} $ fails to satisfy (G1) and (G1s).

Finally, consider $3\le m_1+m_2\le m$ with $ 2 \le m_1 $,
$\vec{\mathbb{X'}} \in L^{0}(\mathbb{I}^{m_1-1}) $ and
$ (\vec{\mathbb{X}},\vec{\mathbb{Y}}) \in \Lina \times \Linb$ such that
$ \mathbb{X'} \cup \{X''\} = \mathbb{X} $, where $ X'' \in \mathbb{X} $ is comonotonic with some element
$ X' \in \mathbb{X'}$.
Without loss of generality, denote by $i$ the position of $X''$ within the vector
$(\vec{\mathbb{X}},\vec{\mathbb{Y}})$.
Then, Lemma \ref{AppendixLemma1} and Example \ref{ExampleMargins} yield
$$
C_{(\vec{\mathbb{X}},\vec{\mathbb{Y}})} (u \, {\bf 1})
	= C_{(\vec{\mathbb{X}},\vec{\mathbb{Y}})} (\boldsymbol{\eta}_{i} (u \, {\bf 1}, {\bf 1}))
	= (T_{\{1,...,m_1+m_2\} \backslash\{i\}} (C_{(\vec{\mathbb{X}},\vec{\mathbb{Y}})})) (u \, {\bf 1})
	= C_{(\vec{\mathbb{X'}},\vec{\mathbb{Y}})} (u \, {\bf 1})
	$$
for every $ u \in \mathbb{I} $ and, hence,
\begin{eqnarray*}
  \big[ C_{(\vec{\mathbb{X}},\vec{\mathbb{Y}})}, M \big]
  & = & \int_{\mathbb{I}^{m_1+m_2}} C_{(\vec{\mathbb{X}},\vec{\mathbb{Y}})} (\uuu) \; \mathrm{d} Q^M (\uuu)
	\\
	& = & \int_{\mathbb{I}} C_{(\vec{\mathbb{X}},\vec{\mathbb{Y}})} (u \, {\bf 1}) \; \mathrm{d} \leb (u)
	= \int_{\mathbb{I}} C_{(\vec{\mathbb{X'}},\vec{\mathbb{Y}})} (u \, {\bf 1}) \; \mathrm{d} \leb (u)
	\\
	& = & \int_{\mathbb{I}^{m_1-1+m_2}} C_{(\vec{\mathbb{X'}},\vec{\mathbb{Y}})} (\uuu) \; \mathrm{d} Q^M (\uuu)
	= \big[ C_{(\vec{\mathbb{X'}},\vec{\mathbb{Y}})}, M \big]
\end{eqnarray*}
and Lemma \ref{AppendixLemma1} together with Example \ref{ExampleMargins} implies
$$
  [C_{(\vec{\mathbb{X}},\vec{\mathbb{Y}})},C_{(\vec{\mathbb{X}},\vec{\mathbb{Y}})}]
	= \big[ T_{\{1,...,m_1+m_2\} \backslash\{i\}} \big( C_{(\vec{\mathbb{X}},\vec{\mathbb{Y}})} \big), T_{\{1,...,m_1+m_2\} \backslash\{i\}} \big( C_{(\vec{\mathbb{X}},\vec{\mathbb{Y}})} \big) \big]
	= [C_{(\vec{\mathbb{X'}},\vec{\mathbb{Y}})},C_{(\vec{\mathbb{X'}},\vec{\mathbb{Y}})}]
$$
Thus,
\begin{eqnarray*}
  d^{m_1,m_2}_{\beta} (\vec{\mathbb{X}},\vec{\mathbb{Y}})
	& = & d^{m_1-1,m_2}_{\beta} (\vec{\mathbb{X'}},\vec{\mathbb{Y}})
	\\
	d^{m_1,m_2}_{\phi} (\vec{\mathbb{X}},\vec{\mathbb{Y}})
	& = & d^{m_1-1,m_2}_{\phi} (\vec{\mathbb{X'}},\vec{\mathbb{Y}})
	\\
	d^{m_1,m_2}_{\tau} (\vec{\mathbb{X}},\vec{\mathbb{Y}})
	& = & d^{m_1-1,m_2}_{\tau} (\vec{\mathbb{X'}},\vec{\mathbb{Y}})
\end{eqnarray*}
which implies that $ d_{\beta} $, $ d_{\phi} $ and $ d_{\tau} $ satisfy (G2).
In Example \ref{HierarchicalMOCCounterExRho2} we show that the extended dissimilarity function $ d_{\rho} $ fails to satisfy (G2).
\end{proof}

The following example shows that the extended dissimilarity functions
$ d_{\beta} $, $ d_{\phi} $ and $ d_{\tau} $ satisfy (G1s).

\begin{example}\label{HierarchicalMOCCounterEx}
Consider $m\geq 3 $, $ m_1 = m_2 = m_3 = 1 $ and the product copula $ \Pi $.
Then, for every random vector $ (X',X'',Y) \in L^{0}(\mathbb{I})  \times L^{0}(\mathbb{I}) \times L^{0}(\mathbb{I})$ having copula
$ T_{\{1,2,3\}} (\Pi) $, the extended dissimilarity functions
$ d_{\beta} $, $ d_{\phi} $ and $ d_{\tau} $ satisfy
\begin{center}
\begin{tabular}{lclcl}
  $d_{\beta}^{1,1} (X',X'')$
	& $\leq$ & $\min\big\{ d_{\beta}^{1,1} (X',Y), d_{\beta}^{1,1} (X'',Y) \big\}$
	&   $<$  & $d^{1+1,1}_{\beta} (\vec{\mathbb{X}},Y)$	
	\\
	$d_{\phi}^{1,1} (X',X'')$
	& $\leq$ & $\min\big\{ d_{\phi}^{1,1} (X',Y), d_{\phi}^{1,1} (X'',Y) \big\}$
	&   $<$  & $d^{1+1,1}_{\phi} (\vec{\mathbb{X}},Y)$	
	\\
	$d_{\tau}^{1,1} (X',X'')$
	& $\leq$ & $\min\big\{ d_{\tau}^{1,1} (X',Y), d_{\tau}^{1,1} (X'',Y) \big\}$
	&   $<$  & $d^{1+1,1}_{\tau} (\vec{\mathbb{X}},Y)$	
\end{tabular}
\end{center}
where $ \mathbb{X} = X' \cup X'' $.
Indeed, we have
\begin{center}
\begin{tabular}{lclclclcl}
  $d_{\beta}^{1,1} (X',X'')$
	&   $=$  & $\min\big\{ d_{\beta}^{1,1} (X',Y), d_{\beta}^{1,1} (X'',Y) \big\}$
	&   $=$  & $ 2/8 $
	&   $<$  & $ 3/8 $
	&   $=$  &$d^{1+1,1}_{\beta} (\vec{\mathbb{X}},Y)$	
	\\
	$d_{\phi}^{1,1} (X',X'')$
	&   $=$  & $\min\big\{ d_{\phi}^{1,1} (X',Y), d_{\phi}^{1,1} (X'',Y) \big\}$
	&   $=$  & $ 2/12 $
	&   $<$  & $ 3/12 $
	&   $=$  & $d^{1+1,1}_{\phi} (\vec{\mathbb{X}},Y)$	
	\\
	$d_{\tau}^{1,1} (X',X'')$
	&   $=$  & $\min\big\{ d_{\tau}^{1,1} (X',Y), d_{\tau}^{1,1} (X'',Y) \big\}$
	&   $=$  & $ 2/8 $
	&   $<$  & $ 3/8 $
  &   $=$  & $d^{1+1,1}_{\tau} (\vec{\mathbb{X}},Y)$	
\end{tabular}
\end{center}
Thus, the extended dissimilarity functions
$ d_{\beta} $, $ d_{\phi} $ and $ d_{\tau} $ satisfy (G1s).
\end{example}

We conclude this section by showing that the extended dissimilarity function $ d_{\rho} $ fails to satisfy (G1), (G1s) and (G2).

\begin{example}\label{HierarchicalMOCCounterExRho1}
Consider $m\geq 6 $, $ m_1 = m_2 = m_3 = 2 $ and the product copula $ \Pi $.
Then, for every random vector $ (\vec{\mathbb{X'}},\vec{\mathbb{X''}},\vec{\mathbb{Y}}) \in L^{0}(\mathbb{I}^2) \times L^{0}(\mathbb{I}^2) \times L^{0}(\mathbb{I}^2) $ having copula
$ T_{\{1,\dots,6\}} (\Pi)$, the extended dissimilarity function
$ d_{\rho} $ satisfies
$$
d_\rho^{2,2} (\vec{\mathbb{X'}},\vec{\mathbb{X''}}) = \min\big\{ d_{\rho}^{2,2} (\vec{\mathbb{X'}},\vec{\mathbb{Y}}), d_{\rho}^{2,2} (\vec{\mathbb{X''}},\vec{\mathbb{Y}}) \big\}
>  d^{2+2,2}_{\rho} (\vec{\mathbb{X}},\vec{\mathbb{Y}})	
$$
where $ \mathbb{X} = \vec{\mathbb{X'}} \cup \vec{\mathbb{X''}} $.
Indeed, we have
$$
d_\rho^{2,2} (\vec{\mathbb{X'}},\vec{\mathbb{X''}})
=\min\big\{ d_{\rho}^{2,2} (\vec{\mathbb{X'}},\vec{\mathbb{Y}}), d_{\rho}^{2,2} (\vec{\mathbb{X''}},\vec{\mathbb{Y}}) \big\}
=\frac{616}{4480}>\frac{570}{4480}
=d^{2+2,2}_{\rho} (\vec{\mathbb{X}},\vec{\mathbb{Y}})
$$	
Thus, the extended dissimilarity function
$ d_{\rho} $ fails to satisfy (G1) and also (G1s).
\end{example}

\begin{example}\label{HierarchicalMOCCounterExRho2}
Consider $m \geq 4 $, $ m_1 = 2 $, $ m_2 = m-2 $
and the copula $ C: \mathbb{I}^{2} \times \mathbb{I}^{m-2} \to \mathbb{I} $ given by
$$
  C(\uuu,\vvv)
	:= M(\uuu)\, \Pi(\vvv)
$$
(see e.g. \cite{DurSem16}). Then,
for every random vector
$ \vec{\mathbb{X}} = (X_1,X_2)^{\prime} \in L^{0}(\mathbb{I}^{2}) $
and every random vector $ \vec{\mathbb{Y}} \in L^{0}(\mathbb{I}^2) $
such that $ (\vec{\mathbb{X}},\vec{\mathbb{Y}}) $ has copula $ T_{\{1,2,3,4\}} (C) $ and hence $X_1$ and $X_2$ are comonotonic,
the extended dissimilarity function $ d_{\rho} $ satisfies
$$
  d^{2,2}_{\rho} (\vec{\mathbb{X}},\vec{\mathbb{Y}})
	\neq d^{1,2}_{\rho} (X_1,\vec{\mathbb{Y}})
$$
Indeed, we obtain
$
  d^{2,2}_{\rho} (\vec{\mathbb{X}},\vec{\mathbb{Y}})
	  =  \frac{14}{120}
	\neq \frac{15}{120}
	  =  d^{1,2}_{\rho} (X_1,\vec{\mathbb{Y}}).
$
Thus, the extended dissimilarity function
$ d_{\rho} $ fails to satisfy (G2).
\end{example}

\begin{remark}{}
The choice of some normalizing constants in the definition of dissimilarity functions based on measures of association is crucial. In particular, the direct use of multivariate versions of these measures may be flawed in some cases, as the following example indicates.
\\
Consider the measure of concordance Kendall's tau $\kappa$ (see, e.g., \cite{Genest2011}) and,
for $ m_1, m_2 \in \mathbb{N} $ with $ 2 \leq m_1 + m_2 \leq m $,
define the map
$d^{m_1,m_2}_{\kappa} \colon \Lina\times\Linb\to[0,+\infty[ $ by letting
$$
  d^{m_1,m_2}_{\kappa} (\vec{\mathbb{X}},\vec{\mathbb{Y}})
	:= 1 - \kappa (\vec{\mathbb{X}},\vec{\mathbb{Y}})
	 = \frac{\big[M,M\big] - \big[ C_{(\vec{\mathbb{X}},\vec{\mathbb{Y}})}, C_{(\vec{\mathbb{X}},\vec{\mathbb{Y}})} \big]}{\big[ M, M \big] - \big[\Pi,\Pi\big]}
	 = \frac{d^{m_1,m_2}_{\tau} (\vec{\mathbb{X}},\vec{\mathbb{Y}})}{\big[ M, M \big] - \big[\Pi,\Pi\big]}
$$
Then (G1) is equivalent to the inequality
$
  \kappa (\vec{\mathbb{X}},\vec{\mathbb{Y}})
	\leq \max\big\{ \kappa (\vec{\mathbb{X'}},\vec{\mathbb{Y}}), \kappa (\vec{\mathbb{X''}},\vec{\mathbb{Y}}) \big\}
$
for all $3\le m_{1}+m_{2}+m_3\le m$ and for all
$ (\vec{\mathbb{X'}}, \vec{\mathbb{X''}}, \vec{\mathbb{Y}}) \in L^{0}(\mathbb{I}^{m_{1}}) \times L^{0}(\mathbb{I}^{m_{2}}) \times L^{0}(\mathbb{I}^{m_{3}}) $,
	where $ \mathbb{X} := \mathbb{X'} \cup \mathbb{X''} $
and $ \mathbb{X'} $, $\mathbb{X''} $ and $\mathbb{Y} $ are pairwise disjoint.
Now, consider $m\geq 4 $ and the copula $ C: \mathbb{I}^{m} \to \mathbb{I}$ given by
$$
  C(\uuu)
	:= \Pi(\uuu) + \prod_{i=1}^{m} u_i \prod_{i=1}^{4} (1-u_i)
$$
Then, for every random vector
$ (X',X'',\vec{\mathbb{Y}}) \in L^{0}(\mathbb{I}) \times L^{0}(\mathbb{I}) \times L^{0}(\mathbb{I}^2) $
having copula
$ T_{\{1,2,3,4\}} (C) $, the above inequality reduces to
$
  \kappa (\vec{\mathbb{X}},\vec{\mathbb{Y}})
	\leq 0
$
where $ \mathbb{X} := X' \cup X'' $.
However, straightforward calculation yields
$
  \kappa (\vec{\mathbb{X}},\vec{\mathbb{Y}}) = \frac{2}{567}
$
which contradicts (G1).
Thus, although $d^{m_1,m_2}_{\tau}$ satisfies (G1), $d^{m_1,m_2}_{\kappa}$ may fail to satisfy (G1).
\end{remark}


\subsection{Extended dissimilarity functions based on multivariate  tail dependence}\label{subsec:TDC}

In this section we study an extended dissimilarity function based on a modified version of the classical lower tail dependence coefficient (see, e.g., \cite{DurSem16}). This kind of dissimilarity concept is useful in order to  detect different tail association in random vectors. In the literature, similar concepts have been considered for the analysis of financial time series. See, e.g., \cite{DeLZuc11,DurPapTor15StatPap,Jietal18,Yang2018}.

For $ m_1, m_2 \in \mathbb{N} $ with $ 2 \leq m_1 + m_2 \leq m $,
we define the function $ d^{m_1,m_2}_{\rm LTD} \colon \Lina\times\Linb\to[0,+\infty[ $ by letting
$$
  d^{m_1,m_2}_{\rm LTD} (\vec{\mathbb{X}},\vec{\mathbb{Y}})
	:= 1 - \limsup_{u \to 0^+} \; \frac{C_{(\vec{\mathbb{X}},\vec{\mathbb{Y}})} (u {\bf 1})}{u}
$$
Notice that, provided that the above limit superior coincides with the limit inferior, then $d^{1,1}=1-\lambda_L$, where $\lambda_L$
is the lower tail dependence coefficient of $(X,Y)$.

\begin{theorem}
For all $ 2 \leq m_1 + m_2 \leq m $, $ d^{m_1,m_2}_{\rm LTD} $ is a $(m_1,m_2)$-dissimilarity function,
and thus,
$$
  d_{\rm LTD}: \bigcup_{2 \leq m_1+m_2 \leq m} \Lina\times\Linb\to[0,+\infty[
$$
with $ d_{\rm LTD}|_{\Lina\times\Linb} := d^{m_1,m_2}_{\rm LTD} $ is an extended dissimilarity function satisfying (L1), (L1c), (G1), (G1s) and (G2).
\end{theorem}

\begin{proof}
It is straightforward to show that $ d_{\rm LTD} $ is an extended dissimilarity function satisfying (L1) and (L1c).

Now, consider
$ 3 \le m_{1}+m_{2}+m_3 \le m $,
the random vector
$ (\vec{\mathbb{X'}}, \vec{\mathbb{X''}}, \vec{\mathbb{Y}})
  \in L^{0}(\mathbb{I}^{m_{1}}) \times L^{0}(\mathbb{I}^{m_{2}}) \times L^{0}(\mathbb{I}^{m_{3}}) $
such that \eqref{eq:reducibility1} holds,
$ \mathbb{X'} $, $\mathbb{X''} $ and $\mathbb{Y} $ are pairwise disjoint  and put $ \mathbb{X} := \mathbb{X'} \cup \mathbb{X''} $.
Then
\begin{eqnarray*}
  C_{(\vec{\mathbb{X}},\vec{\mathbb{Y}})} (u {\bf 1})
  & \leq & C_{(\vec{\mathbb{X}},\vec{\mathbb{Y}})} \big( \boldsymbol{\eta}_{\{m_1+1,...,m_1+m_2\}} (u {\bf 1}, {\bf 1}) \big)
	\\
	& 	=  & \big( T_{\{1,...,m_1+m_2+m_3\} \backslash \{m_1+1,...,m_1+m_2\}} \big( C_{(\vec{\mathbb{X}}, \vec{\mathbb{Y}})} \big) \big) (u {\bf 1})
	\\
	&	  =  & C_{(\vec{\mathbb{X'}},\vec{\mathbb{Y}})} (u {\bf 1})
  \\
  C_{(\vec{\mathbb{X}},\vec{\mathbb{Y}})} (u {\bf 1})
  & \leq & C_{(\vec{\mathbb{X}},\vec{\mathbb{Y}})} \big( \boldsymbol{\eta}_{\{1,...,m_1\}} (u {\bf 1}, {\bf 1}) \big)
	\\
	&   =  & \big( T_{\{m_1+1,...,m_1+m_2+m_3\}} \big( C_{(\vec{\mathbb{X}}, \vec{\mathbb{Y}})} \big) \big) (u {\bf 1})
	\\
	&   =  & C_{(\vec{\mathbb{X''}},\vec{\mathbb{Y}})} (u {\bf 1})
\end{eqnarray*}
for every $ u \in \mathbb{I} $, and thus,
$ \min\big\{ d^{m_{1},m_3}_{\rm LTD} (\vec{\mathbb{X'}},\vec{\mathbb{Y}}),
						 d^{m_{2},m_3}_{\rm LTD} (\vec{\mathbb{X''}},\vec{\mathbb{Y}}) \big\}
  \leq d^{m_{1}+m_{2},m_3}_{\rm LTD} (\vec{\mathbb{X}},\vec{\mathbb{Y}}).$
This proves (G1).
In Example \ref{HierarchicalLTDCounterEx} we show that the extended dissimilarity function $ d_{\rm LTD} $ satisfies (G1s).

Finally, consider $3\le m_1+m_2\le m$ with $ 2 \le m_1 $,
$\vec{\mathbb{X'}} \in L^{0}(\mathbb{I}^{m_1-1}) $ and
$ (\vec{\mathbb{X}},\vec{\mathbb{Y}}) \in \Lina \times \Linb$ such that
$ \mathbb{X'} \cup \{X''\} = \mathbb{X} $, where $ X'' \in \mathbb{X} $ is comonotonic with some element
$ X' \in \mathbb{X'}$.
Without loss of generality, denote by $i$ the position of $X''$ within the vector
$(\vec{\mathbb{X}},\vec{\mathbb{Y}})$.
Then Lemma \ref{AppendixLemma1} yields
$ C_{(\vec{\mathbb{X}},\vec{\mathbb{Y}})} (u {\bf 1})
	= C_{(\vec{\mathbb{X}},\vec{\mathbb{Y}})} (\boldsymbol{\eta}_{i} (u {\bf 1}, {\bf 1}))
	= (T_{\{1,...,m_1+m_2\} \backslash \{i\}} (C_{(\vec{\mathbb{X}},\vec{\mathbb{Y}})})) (u {\bf 1})
	= C_{(\vec{\mathbb{X'}},\vec{\mathbb{Y}})} (u {\bf 1}) $
for every $ u \in \mathbb{I} $ and hence
$
  d^{m_1,m_2}_{\rm LTD} (\vec{\mathbb{X}},\vec{\mathbb{Y}})
	= d^{m_1-1,m_2}_{\rm LTD} (\vec{\mathbb{X'}},\vec{\mathbb{Y}}).
$
This proves (G2) and, hence, the assertion.
\end{proof}

The following example shows that the extended dissimilarity function $ d_{\rm LTD} $ satisfies (G1s).

\begin{example}\label{HierarchicalLTDCounterEx}
Consider $m\geq 3 $, $ m_1 = m_2 = m_3 = 1$, and the copula
$ C: \mathbb{I}^{m} \to \mathbb{I}$ given by
$$
  C(\uuu)
	:= \left( \sum_{i=1}^{m} u_i^{-1/2} - (n-1) \right)^{-2}
$$
which is the Clayton copula with parameter $1/2$.
Then, for every random vector
$ (X',X'',Y) \in L^{0}(\mathbb{I}) \times L^{0}(\mathbb{I}) \times L^{0}(\mathbb{I}) $
having copula $ T_{\{1,2,3\}} (C) $,
the extended dissimilarity function $ d_{\rm LTD} $ satisfies
$$
d_{\rm LTD}^{1,1} (X',X'')
\leq \min\big\{ d_{\rm LTD}^{1,1} (X',Y), d_{\rm LTD}^{1,1} (X'',Y) \big\}
<d^{1+1,1}_{\rm LTD} (\vec{\mathbb{X}},Y)
$$	
where $ \mathbb{X} = X' \cup X''$. Indeed, we have
$$
d_{\rm LTD}^{1,1} (X',X'')
=\min\big\{ d_{\rm LTD}^{1,1} (X',Y), d_{\rm LTD}^{1,1} (X'',Y) \big\}
=\dfrac{3}{4}<\dfrac{8}{9} 
=d^{1+1,1}_{\rm LTD} (\vec{\mathbb{X}},Y).
$$
Thus, the extended dissimilarity function $ d_{\rm LTD} $ satisfies (G1s).
\end{example}

\begin{remark} \label{TDC.rem} \leavevmode
\begin{enumerate}
\item[(1)] Notice that $d_{\rm LTD}$ does not satisfy (L2), since lower and upper tail behaviour of a copula may be different.

\item[(2)] Note also that $d_{\rm LTD}$ does not satisfy (L3).
To this end, consider, for instance, the bivariate copula $C_k$ that is an ordinal sum of $(M,\Pi)$ with respect to $([0,1/k],[1/k,1])$ (see, e.g., \cite{DurSem16}). Then $C_k$ tends to $\Pi$, as $k$ tends to $+\infty$  with $d^{1,1}_{\rm LTD}(\Pi)=1$, but $d^{1,1}_{\rm LTD}(C_k)=0$ for every $k\ge 2$.

\item[(3)] 
In particular then also $d_{\rm LTD}^{1,1} $ does not satisfy (L3)
and this hence transfers to the extended dissimilarity functions based on linkage methods $ d_{\rm min} $, $ d_{\rm ave} $ and $ d_{\rm max} $ (compare Theorem \ref{MoDLinkage}).
\end{enumerate}
\end{remark}

\section{The hierarchical clustering procedure}\label{sec:algorithm}

Here, we summarize how a general agglomerative hierarchical algorithm based on extended dissimilarity functions can be implemented (see, for instance, \cite{Eve11,Gor87,Koj04}). To this end, we remind that we aim at determining a suitable partition of the (finite) set $\mathcal{X}=\{X_1,\dots,X_m\}$ of $m\ge 3$ continuous r.v.'s into non-empty and non-overlapping classes. 

%

Given a dissimilarity index $\widetilde{d}$ induced by some extended dissimilarity function
$$
d:
	\bigcup_{2 \leq m_1+m_2 \leq m} \Lina\times\Linb\to[0,+\infty[
$$
the different steps of an agglomerative hierarchical clustering algorithm based on $d$ are given below:
\begin{enumerate}
\item[(1)] Each object of $\mathcal{X}$ forms a class.
\item[(2)] For each pair of classes $\mathbb{X}$ and $\mathbb{Y}$, one computes $\widetilde{d}(\mathbb{X},\mathbb{Y})$.
\item[(3)]  A pair of classes having the smallest dissimilarity degree, say $\{\mathbb{X}_1,\mathbb{Y}_1\}$, is identified, then the composite class $\mathbb{X}_1\cup\mathbb{Y}_1$ is formed and the number of classes is decremented.
\item[(4)]  Steps (2), (3) and (4) are repeated until the number of classes is equal to $1$.
\end{enumerate}

The hierarchy of classes built by the clustering algorithm can be hence represented by means of a \emph{dendrogram}, from which a suitable partition of $\mathcal{X}$ can be derived (see, for instance, \cite{Eve11}).

\bigskip
Now, while these steps are common to any agglomerative algorithm, the use of the extended dissimilarity function may provide some important insights into the agglomerative hierarchical algorithm. In fact, the procedure can use: either (a) the information about the pairwise dependence, as in the dissimilarity function based on linkage methods; or (b) the information about their global (higher dimensional) copula. The latter method, in particular, will allow us to
detect those dependencies that only appear in higher dimensions, a feature that can be quite appealing in applications.

\begin{example}\label{Ex:dendro_HAC}
As an illustrative example, consider a set $\mathcal{X}$ formed by $6$ pairwise independent r.v.'s such that $\mathbb{X}=\{X_1,X_2,X_3\}$ and $\mathbb{Y}=\{Y_1,Y_2,Y_3\}$ are, respectively, globally dependent. For instance, we may assume that they are coupled with a trivariate FGM copula with parameter $\theta_1$ and $\theta_2$, respectively $(\theta_1>\theta_2)$.
\\
Now, every dissimilarity index based on classical linkage methods cannot recognize the difference among the two groups and its related dendrogram would be similar to the representation in Figure \ref{fig:dendro_HAC} (left). However, if we consider the dissimilarity index based on multivariate Kendall's tau, then the procedure could produce a different output and recognize the dendrogram structure as in Figure \ref{fig:dendro_HAC} (right).
\end{example}

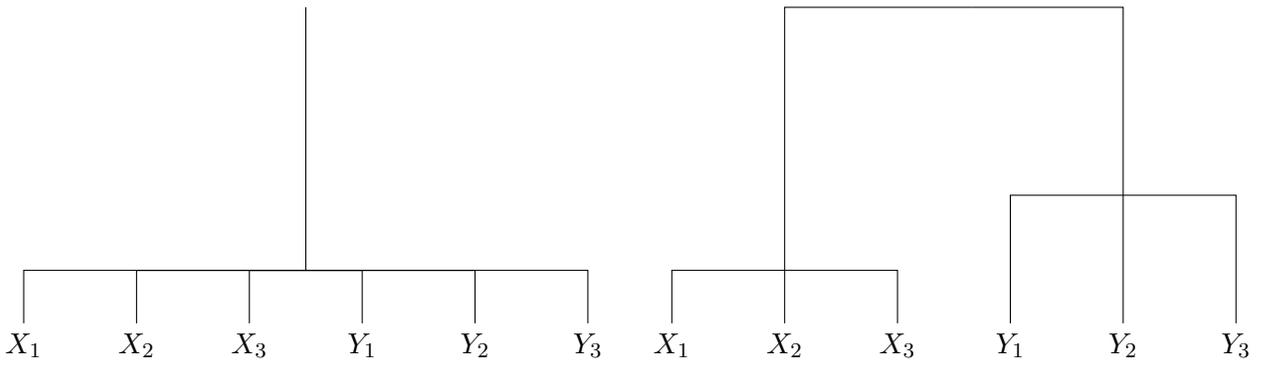
\begin{figure}
\begin{tabular}{lr}
\begin{tikzpicture}[sloped]
\node (a) at (-5,0) {$X_1$};
\node (b) at (-3.5,0) {$X_2$};
\node (c) at (-2,0) {$X_3$};
\node (d) at (-0.5,0) {$Y_1$};
\node (e) at (1,0) {$Y_2$};
\node (h) at (2.5,0) {$Y_3$};
\node (all) at (-1.25,1) {};

\draw  (a) |- (all.center);
\draw  (b) |- (all.center);
\draw  (c) |- (all.center);
\draw  (d) |- (all.center);
\draw  (e) |- (all.center);
\draw  (h) |- (all.center);
\draw (-1.25,1)--(-1.25,4.5);
\end{tikzpicture}

&

\begin{tikzpicture}[sloped]
\node (a) at (-5,0) {$X_1$};
\node (b) at (-3.5,0) {$X_2$};
\node (c) at (-2,0) {$X_3$};
\node (d) at (-0.5,0) {$Y_1$};
\node (e) at (1,0) {$Y_2$};
\node (h) at (2.5,0) {$Y_3$};
\node (abc) at (-3.5,1) {};
\node (deh) at (1,2) {};
\node (all) at (-1,4.5) {};

\draw  (a) |- (abc.center);
\draw  (b) |- (abc.center);
\draw  (c) |- (abc.center);
\draw  (d) |- (deh.center);
\draw  (e) |- (deh.center);
\draw  (h) |- (deh.center);
\draw  (abc.center) |- (all.center);
\draw  (deh.center) |- (all.center);
\end{tikzpicture}
\end{tabular}
\caption{Two illustrative examples of dendrogram representation of a random vector based on different extended dissimilarity functions. See Example \ref{Ex:dendro_HAC}.}\label{fig:dendro_HAC}
\end{figure}

Apart from the case when the probability law of the $\mathcal{X}=\{X_1,\dots,X_m\}$ ($m\ge 3$) is known (i.e. by some fitting procedures and/or expert opinion), the information about $\mathcal{X}$ is usually recovered from some available observations, which can be considered as random sample from $X_1,\dots,X_m$, denoted by $(x_{ij})$ with $i=1,\dots,n$ and $j=1,\dots,m$.
In such a case, depending on the dissimilarity functions, specific estimation procedures should be adopted.

\begin{example}\label{Ex:plug-in}
Consider the case when a dissimilarity function $d^{m_1,m_2}$ can be expressed as a smooth function of a given measure of association $\mu$ for $(m_1+m_2)$-dimensional random vectors, say
$$
d^{m_1,m_2}=f(\mu)
$$
Then, in view of a suitable application of continuous mapping theorem, a (plug--in) estimator of $d^{m_1,m_2}$ is given by
$$
\widehat{d}^{m_1,m_2}=f(\widehat\mu),
$$
where $\widehat \mu$ is a convenient estimator of $\mu$. Such a procedure can be, for instance, applied to the dissimilarity functions considered in Section \ref{SubsectionExtDissMOD}.
\end{example}

\begin{remark}
In the case of multivariate time series, i.e. when, for a fixed $j$, $(x_{ij})$ presents a time--varying behaviour, it is common to apply the hierarchical algorithm not directly on the historical time series, but on the residual time series obtained after fitting each univariate time series with an appropriate time-varying model (like ARMA-GARCH specification). Such a general framework is described, for instance, in \cite{Pat12} (see also \cite{Acaretal19}) and applied, among others, in \cite{DeLZuc11,DurPapTor14ADAC}.
\end{remark}

In general, every dissimilarity function has a strong impact on the clustering procedure, since each one can have a quite different interpretation. However, when extended dissimilarity functions based on different linkage methods are considered, it would be convenient to compare them since they are defined from the same bivariate dissimilarity function.  Below, via a simulation study, we check whether the choice of the linkage method may have a relevant impact on the performance of the algorithm.

\subsection{A simulation study about linkage methods}\label{subsec:simula}

Here, we compare the performance of  hierarchical clustering methods where the extended dissimilarity functions are based on average, single and complete linkage method, while the pairwise dissimilarities are obtained from
$d^{1,1}_\beta$, $d^{1,1}_\phi$, $d^{1,1}_{\rho}$ and $d^{1,1}_{\tau}$. Notice that, since these pairwise dissimilarities are based on classical measures of association, their (non-parametric) estimation is grounded on the (classical) empirical versions of these measures, as described in \cite{Genetal13JSFdS,Sch_et_al10CTA,Ube05} among others (see also Example \ref{Ex:plug-in}).

First, we consider the following setup. A random vector $\vec{\mathbb{X}}$ of dimension $m=15$ is constructed in the following way:
\begin{itemize}
\item the random vector is formed by three independent subvectors, say $(\vec{\mathbb{X}}_1,\vec{\mathbb{X}}_2,\vec{\mathbb{X}}_3)$;
\item the dimension $m_i$ of each $\vec{\mathbb{X}}_i$ is randomly chosen from $2$ to $11$ to ensure that each group has $2$ elements and $m_1+m_2+m_3=m$;
\item each $\vec{\mathbb{X}}_i$ is distributed according to a copula generated from four different copula models, namely Clayton, Frank, Gumbel and equicorrelated Gaussian (for the definition of these families, see, e.g., \cite{DurSem16}), with pairwise Kendall's tau equal to $\tau$.
\end{itemize}
For $B=500$ replications, the simulation study is then performed simulating $N$ independent realizations $(N\in\{50,100,250\})$ from $\vec{\mathbb{X}}$ with $\tau\in\{0.1,0.2,0.3\}$.
Hence, for each simulated scenario the Adjusted Rand Index \cite{HubAra85} (ARI, hereafter) is calculated to measure the agreement between the obtained partition and the true one. Here, the partition is obtained by cutting the dendrogram so that three groups are derived.

The distribution of ARI for each scenario is shown in Figures~\ref{Clayton},~\ref{Frank},~\ref{Gumbel} and~\ref{Gaussian}.
As the results for the four copula models are very similar, we only comment those obtained for the Clayton copula and shown in Fig.~\ref{Clayton}. We remind that a larger Adjusted Rand Index means a higher agreement between two partitions and the maximum value of the index is $1$. As one could have expected, the lower is the degree of dependence among the variables of a group, the harder is for the hierarchical clustering algorithm to identify the true partition. Moreover, the larger is the sample size, the better are the results for a given dependence degree. As far as linkage methods are concerned, one can see that, remarkably differences appear only when the dependence level is really low, i.e. lower than $0.3$. In these cases, irrespectively from the dissimilarity measures, the average linkage method appears to be more satisfactory than the complete and the single ones. As for the pairwise dissimilarity function, $d^{1,1}_\beta$ appears to be the worst choice in case of weak dependence among groups. Overall, the average linkage performed the best, which confirms its potential frequently proved in the literature \cite{Hall68,Eisen98,Bottegoni06,Kumar12}, especially when $d^{1,1}_\phi$, $d^{1,1}_{\rho}$ and $d^{1,1}_{\tau}$ are used.

\par
Since the choice of the copula family seems to be irrelevant in the previous simulation, we fix one specific family, namely Clayton class, and perform a similar simulation study in higher dimensions. Specifically, we consider a random vector $\vec{\mathbb{X}}$ of dimension $m\in\{60,120\}$ such that:
\begin{itemize}
\item the random vector is formed by $K\in \{6,10\}$ independent subvectors, say $\vec{\mathbb{X}}_i$;
\item the dimension $m_i$ of each $\vec{\mathbb{X}}_i$ is $m/K$;
\item each $\vec{\mathbb{X}}_i$ is distributed according to a Clayton copula with pairwise Kendall's tau equal to $\tau$.
\end{itemize}
For $B=500$ replications, the simulation study is then performed simulating $N$ independent realizations $(N\in\{100,250\})$ from $\vec{\mathbb{X}}$ with $\tau\in\{0.1,0.2,0.3\}$.

The results can be seen from Figure~\ref{100_3} to Figure~\ref{250_6}. Summarizing, both for $K=6$ and $K=10$, the average linkage performs better than the other two linkages, while the single linkage is the worst one. The complete linkage shows a performance similar to the average linkage when $\tau>0.1$. There are no remarkable differences among dissimilarities by varying $m$ in $\{60,120\}$ and the slight differences are remarkably reduced as when $\tau>0.1$ and $N=250$, cases where all the measures show an almost perfect performance (except for the single linkage and $d_\beta^{1,1}$).


\section{Applications}\label{sec:appli}

In order to show the ability of our methodology in the statistical practice, we present some empirical analysis.

\subsection{Analysis of gene expressions}

First, we focus on the NCI60 data set which is available in the R package \texttt{made4} \cite{MADE4} and contains 144 gene expression (log-ratio measurements) rows and 60 cell line columns. Gene expressions have been extracted by using the cDNA spotted microarray technology \cite{Ross2000} and pre-processed as described by \cite{Culhane2003}.
The study has been carried out by the National Cancer Institute's (NCI) Developmental Therapeutics Program (DTP) and human tumour cell lines have been derived from patients with leukaemia (LEUK), melanoma (MELAN), non-small colon lung (NSCLC), colon (COLON), central nervous system (CNS), ovarian (OVAR), renal (RENAL), breast (BREAST) and prostate (PROSTATE) cancers. Here, we divided the human tumour cell lines in two groups according to the (bivariate) Kendall's $\tau$ and using $0.3$ as cut-off. Precisely, for each subset of human tumour cell lines, i.e. for each kind of tumour, the pairwise Kendall's $\tau$ correlation matrix has been computed. If at least 60\% (59.4\% for the BREAST cancer group) of pairwise correlation coefficients is greater than or equal to 0.3, then that kind of tumour has been considered as `tumour with high dependence', otherwise it has been classified as `tumour with low dependence'. The rationale is to show empirical results comparable with the scenarios simulated in the performed Monte Carlo studies. Hence, as for the tissues with low dependence ($\tau<0.3$) we have 8 BREAST, 9 NSCLC, and 6 OVAR, while as for the tissues with high dependence ($\tau\geq0.3$) we have 6 CNS, 7 COLON, 6 LEUK, 8 MELAN, 2 PROSTATE, and 8 RENAL.

Tables~\ref{tab:low} and~\ref{tab:high} show the obtained results. Coherently with the simulation results, when the dependence is low (see Table~\ref{tab:low}), any linkage method, irrespectively of the kind of extended dissimilarity function, is seldom able to recognize the true partition, whereas, when the dependence is mild or high (see Table~\ref{tab:high}), then  the single linkage method appears to perform badly while the average and the complete are very good competitors. Here, one may argue that the global properties and, particularly the reducibility property (G1), can play a role in explaining these performance (see Theorem \ref{LinkageMethodsTheorem}). In addition, we notice that the kind of dissimilarity measure appears to have an impact on the goodness of the final partition only when the average linkage method is used, in this case $d_\rho$ appears to be the best dissimilarity measure.

\bigskip
The second example concerns the data set discussed in \cite{Notterman2001} containing the transcript of 7086 human mRNAs from 4 normal tissues and 4 adenoma tissues. By applying the hierarchical clustering we want to evaluate the capability of distinguishing the two tissue types. In this empirical case, all tissues have a quite high Kendall's $\tau$ correlation ($>0.607$) and high Spearman's $\rho$ correlation ($>0.766$). The resulting clusterings by varying dissimilarity measure and linkage method are shown in Table~\ref{tab:Not} and Fig.~\ref{fig:dendro}.
Note that we are here using the Rand Index \cite{Rand71} instead of its adjusted version since the number and the size of groups are very small.

The dissimilarity measure $d_\beta$ and the complete linkage method is the only combination able to perfectly recognize the two tissue types, thus supporting the concept that genome-wide expression profiling may permit a classification of solid tumors. Again, the effect of the kind of dissimilarity measure appears to be irrelevant.

Here, it is interesting to note that the extended dissimilarity functions $d_\beta$, $d_\phi$ and $d_\tau$ based on measures of multivariate dependence (see Section \ref{SubsectionExtDissMOD}) perfectly group the tissues (Rand index equal to $1$).
For the extended dissimilarity function $d_\rho$, however, the obtained Rand index equals $0.5714$ since two adenoma tissues have been clustered with the normal ones. Thus, one may argue that this performance is due to the fact that $d_\rho$ does not satisfy the reducibility property (G1) (see Theorem \ref{thm:global_d}).

Noteworthy, the dissimilarity $d_\tau$ in its multivariate version is the most computationally heavy measure.

Finally, for the sake of illustration, we discuss the steps of the clustering procedure by means of the extended dissimilarity function $d_\tau$ (see Section \ref{SubsectionExtDissMOD}).
Table \ref{tab:Kendall} provides the merging steps together with the corresponding values of $d_\tau$ and multivariate Kendall's tau (see, e.g., \cite{Genest2011, fumcsc2019}).

As can be seen from the values of Kendall's tau there is a huge $4$- and also $8$-dimensional dependence between the tissues.
Thus, it seems as if the multivariate versions of $d_\beta$, $d_\phi$ and $d_\tau$ perform entirely satisfactory when the r.v.'s are highly dependent and the sample size is large enough.
Again,
one may also conclude from the values of $d_\tau$ that the reducibility property is crucial.


\subsection{Analysis of financial time series}

Here, we provide an illustration of a copula-based clustering procedure based on financial time series. To this end, we consider the dataset formed by the end-day prices of the 505 constituents of the Standard \& Poor 500 index (S\&P hereafter) observed in the financial crisis of 2007-2008 is analysed, by complementing the analysis performed in \cite{DiLDurPap17RBN}. The dataset is available in the \texttt{R} package \texttt{qrmdata} \cite{qrmdata}, where the data are classified according to the Global Industry Classification Standard sector information. We consider 756 daily log-returns recorded from 2007-01-01 to 2009-12-31 on 461 constituents which have not missing data and belong to the following sectors (the number of companies in each sector is in parenthesis): Consumer discretionary (77), Consumer staples (33), Energy (36), Financials (84), Health care (51), Industrials (62), Information technology (59), Materials (25), Telecommunications services (5), and Utilities (29).

Following the copula-based approach for the analysis of time series (see, e.g., \cite{Pat12}), we fit a suitable marginal model to each of the 461 constituents to remove serial dependence. In particular, based on~\cite{DiLDurPap17RBN} we adopt the ARMA(1,1)-GARCH(1,1) model with innovations following a Student-\textit{t} distribution. Once the corresponding residuals have been extracted, hierarchical clustering algorithms are applied by varying dissimilarity measures and linkage methods.

Table~\ref{tab:SP} shows the agreement between the sector classification given by S\&P index and the group composition determined for each considered combination of a dissimilarity measure and a linkage method (here, ARI is used).
As it can be seen, single linkage method shows the worst agreement irrespective from the dissimilarity measure.
On the contrary, the performance of the average and the complete linkage method appear quite different from each other and, on this set of data, the complete linkage outperforms the average linkage. As expected, however, the benchmark group composition provided by sectors reflects poorly the comovements of financial time series.

When we consider a dissimilarity function based on the (pairwise) lower tail dependence coefficient discussed in Section~\ref{subsec:TDC} computed using the nonparametric estimator by~\cite{Schmid07}, the agreement between the obtained group composition and the benchmark sector-wise group composition is even worse. 
In fact, the ARI index equals to 0.003 for the average, 0.001 for the single, and 0.175 for the complete linkage method. In other words, as expected, grouping by economic/financial sectors may not reflect the real comovements of time series, especially in bearish periods.


\section{Conclusions}\label{sec:concl}

We have provided a theoretical foundation for the study of hierarchical clustering algorithms based on (rank-based) dissimilarity measures. To this end, we have introduced and studied dissimilarity functions for continuous random vectors, which are based on the use of copulas. Novel properties of a  dissimilarity have been considered (see Table \ref{tab:diss-prop}) and various dissimilarity measures have been analysed with respect to their main features (see Table \ref{tab:diss-satisfy}).
The obtained results may provide computational and practical insights that may guide for the choice of the most appropriate dissimilarity function for the problem at hand.

Finally, we would like to remark that the simulations and the empirical analysis have been performed in \cite{R}, also by means of the package \cite{Rcopula}.


\section*{Acknowledgements}
{\small We would like to thank the Associate Editor and the anonymous Reviewers for several helpful suggestions that have served to add clarity and breadth to the earlier version of this paper.

FD has been supported by the project ``Stochastic Models for Complex Systems'' by Italian MIUR (PRIN 2017, Project no. 2017JFFHSH). FMLDL has been supported by the project ``The use of Copula for the Analysis of Complex and Extreme Energy and Climate data (CACEEC)'' by the Free University of Bozen-Bolzano, Faculty of Economics and Management (Grant Nos. WW200S).
SF gratefully acknowledges the support of the WISS 2025
project ’IDA-lab Salzburg’ (20204-WISS/225/197-2019 and 0102-F1901166-KZP).}

{\footnotesize

\begin{thebibliography}{10}

\bibitem{Acaretal19}
E.~F. Acar, C.~Czado, and M.~Lysy.
\newblock Flexible dynamic vine copula models for multivariate time series
  data.
\newblock {\em Econometrics and Statistics}, 12:181--197, 2019.

\bibitem{ahnfuchs2018}
J.~Y. Ahn and S.~Fuchs.
\newblock On minimal copulas under the concordance order.
\newblock {\em J. Optim. Theory Appl.}, 184(3):762--780, 2020.

\bibitem{Bonetal04}
G.~Bonanno, G.~Caldarelli, F.~Lillo, S.~Miccich\`e, N.~Vandewalle, and R.~N.
  Mantegna.
\newblock Networks of equities in financial markets.
\newblock {\em Eur. Phys. J. B}, 38(2):363--371, 2004.

\bibitem{Bon19}
A.~Bonanomi, M.~{Nai Ruscone}, and S.~A. Osmetti.
\newblock Dissimilarity measure for ranking data via mixture of copulae.
\newblock {\em Stat. Anal. Data Min.}, 12(5):412--425, 2019.

\bibitem{Bottegoni06}
G.~Bottegoni, A.~Cavalli, and M.~Recanatini.
\newblock A comparative study on the application of hierarchical-agglomerative
  clustering approaches to organize outputs of reiterated docking runs.
\newblock {\em J Chem Inf Model.}, 46(2):852--862, 2006.

\bibitem{CifReg17}
D.~M. Cifarelli and E.~Regazzini.
\newblock On the centennial anniversary of {G}ini's theory of statistical
  relations.
\newblock {\em Metron}, 75(2):227--242, 2017.

\bibitem{CotGen15}
M.-P. C\^{o}t\'e and C.~Genest.
\newblock A copula--based risk aggregation model.
\newblock {\em Canad. J. Statist.}, 43(1):60--81, 2015.

\bibitem{Culhane2003}
A.~C. Culhane, G.~Perri\`ere, and D.~G. Higgins.
\newblock Cross-platform comparison and visualisation of gene expression data
  using co-inertia analysis.
\newblock {\em BMC Bioinformatics}, 21:4--59, 2003.

\bibitem{MADE4}
A.~C. Culhane, J.~Thioulouse, G.~Perri\'ere, and D.~G. Higgins.
\newblock {MADE4: an R package for multivariate analysis of gene expression
  data}.
\newblock {\em Bioinformatics}, 21(11):2789--2790, 2005.

\bibitem{Czaetal12}
C.~Czado, U.~Schepsmeier, and A.~Min.
\newblock Maximum likelihood estimation of mixed {C}-vines with application to
  exchange rates.
\newblock {\em Stat. Model.}, 12(3):229--255, 2012.

\bibitem{DeLZuc11}
G.~{De Luca} and P.~Zuccolotto.
\newblock A tail dependence-based dissimilarity measure for financial time
  series clustering.
\newblock {\em Adv. Data Anal. Classif.}, 5(4):323--340, 2011.

\bibitem{DeLZuc17}
G.~{De Luca} and P.~Zuccolotto.
\newblock A double clustering algorithm for financial time series based on
  extreme events.
\newblock {\em Stat. Risk Model.}, 34(1--2):1--12, 2017.

\bibitem{DeLZuc17b}
G.~{De Luca} and P.~Zuccolotto.
\newblock Dynamic tail dependence clustering of financial time series.
\newblock {\em Stat. Pap.}, 58:641--657, 2017.

\bibitem{Dhaetal02a}
J.~Dhaene, M.~Denuit, M.~J. Goovaerts, R.~Kaas, and D.~Vyncke.
\newblock The concept of comonotonicity in actuarial science and finance:
  theory.
\newblock {\em Insurance Math. Econom.}, 31(1):3--33, 2002.
\newblock 5th IME Conference (University Park, PA, 2001).

\bibitem{DiLDurPap17RBN}
F.~M.~L. Di~Lascio, F.~Durante, and R.~Pappad\`a.
\newblock Copula--based clustering methods.
\newblock In M.~\'Ubeda~Flores, E.~{de Amo}, F.~Durante, and J.~{Fern\'andez
  S\'anchez}, editors, {\em Copulas and {D}ependence {M}odels with
  {A}pplications}, pages 49--67. Springer International Publishing, 2017.

\bibitem{Dilascio16}
F.~M.~L. Di~Lascio and S.~Giannerini.
\newblock Clustering dependent observations with copula functions.
\newblock {\em Stat. Pap.}, 56(3):1--17, 2019.

\bibitem{Disetal17}
M.~Disegna, P.~D'Urso, and F.~Durante.
\newblock Copula-based fuzzy clustering of spatial time series.
\newblock {\em Spat. Stat.}, 21(part A):209--225, 2017.

\bibitem{Disetal13}
J.~Di{\ss}mann, E.~C. Brechmann, C.~Czado, and D.~Kurowicka.
\newblock Selecting and estimating regular vine copulae and application to
  financial returns.
\newblock {\em Comput. Statist. Data Anal.}, 59:52--69, 2013.

\bibitem{DurPapTor14ADAC}
F.~Durante, R.~Pappad\`a, and N.~Torelli.
\newblock Clustering of financial time series in risky scenarios.
\newblock {\em Adv. Data Anal. Classif.}, 8:359--376, 2014.

\bibitem{DurPapTor15StatPap}
F.~Durante, R.~Pappad\`a, and N.~Torelli.
\newblock Clustering of time series via non--parametric tail dependence
  estimation.
\newblock {\em Stat.. Pap.}, 56(3):701--721, 2015.

\bibitem{GDA16}
F.~Durante, G.~Puccetti, M.~Scherer, and S.~Vanduffel.
\newblock Distributions with given marginals: the beginnings. an interview with
  {G}iorgio {Dall'Aglio}.
\newblock {\em Depend. Model.}, 4(1):237--250, 2016.

\bibitem{DurSem16}
F.~Durante and C.~Sempi.
\newblock {\em Principles of copula theory}.
\newblock CRC Press, Boca Raton, FL, 2016.

\bibitem{Eisen98}
M.~B. Eisen, P.~T. Spellman, P.~O. Brown, and D.~Botstein.
\newblock Cluster analysis and display of genome-wide expression patterns.
\newblock {\em PNAS}, 95(25):14863--14868, 1998.

\bibitem{EmbHof13}
P.~Embrechts and M.~Hofert.
\newblock A note on generalized inverses.
\newblock {\em Math. Methods Oper. Res.}, 77(3):423--432, 2013.

\bibitem{Eve11}
B.S. Everitt, S.~Landau, M.~Leese, and D.~Stahl.
\newblock {\em Cluster Analysis}.
\newblock John Wiley \& Sons, Ltd, 5th edition, 2011.

\bibitem{FisNes71}
L.~Fisher and J.~W. {Van Ness}.
\newblock Admissible clustering procedures.
\newblock {\em Biometrika}, 58:91--104, 1971.

\bibitem{fuc2015}
S.~Fuchs.
\newblock {\em Transformations of Copulas and Measures of Concordance}.
\newblock Ph.D. Thesis, Technische Universit{\"a}t Dresden, 2015.

\bibitem{fuc2016a}
S.~Fuchs.
\newblock A biconvex form for copulas.
\newblock {\em Depend. Model.}, 4(1):63--75, 2016.

\bibitem{Fuc16Demo}
S.~Fuchs.
\newblock Copula--{I}nduced {M}easures of {C}oncordance.
\newblock {\em Depend. Model.}, 4(1):205--214, 2016.

\bibitem{fumcsc2019}
S.~Fuchs, Y.~McCord, and K.~D. Schmidt.
\newblock Characterizations of copulas attaining the bounds of multivariate
  {K}endall's tau.
\newblock {\em J. Optim. Theory Appl.}, 178(2):424--438, 2019.

\bibitem{Genetal13JSFdS}
C.~Genest, A.~Carabar\'im-Aguirre, and F.~Harvey.
\newblock Copula parameter estimation using {B}lomqvist's beta.
\newblock {\em J. SFdS}, 154(1):5--24, 2013.

\bibitem{Genest2011}
C.~Genest, J.~Ne{\v s}lehov{\'a}, and N.~Ben~Ghorbal.
\newblock Estimators based on {K}endall's tau in multivariate copula models.
\newblock {\em Aust. N. Z. J. Stat.}, 53:157--177, 2011.

\bibitem{Gib21}
I.~Gijbels, , V.~Kika, and M.~Omelka.
\newblock On the specification of multivariate association measures and their
  behaviour with increasing dimension.
\newblock {\em J. Multivariate Anal.}, 182:104704, 2021.

\bibitem{Gor87}
A.~D. Gordon.
\newblock A review of hierarchical classification.
\newblock {\em J. Roy. Statist. Soc. Ser. A}, 150(2):119--137, 1987.

\bibitem{Goretal17}
J.~G\'orecki, M.~Hofert, and M.~Holen\v{a}.
\newblock Kendall's tau and agglomerative clustering for structure
  determination of hierarchical {A}rchimedean copulas.
\newblock {\em Depend. Model.}, 5(1):75--87, 2017.

\bibitem{Groetal14}
O.~Grothe, J.~Schnieders, and J.~Segers.
\newblock Measuring association and dependence between random vectors.
\newblock {\em J. Multivariate Anal.}, 123, 2014.

\bibitem{Hall68}
A.~V. Hall.
\newblock Methods for showing distinctness and aiding identification of
  critical groups in taxonomy and ecology.
\newblock {\em Nature}, 218(5137):203--204, 1968.

\bibitem{Has09}
T.~{Hastie}, R.~{Tibshirani}, and J.~{Friedman}.
\newblock {\em {The elements of statistical learning. Data mining, inference,
  and prediction.}}
\newblock Springer, New York, NY, 2nd edition, 2009.

\bibitem{Hen16}
C.~{Hennig}, M~{Meila}, F.~{Murtagh}, and R.~{Rocci}, editors.
\newblock {\em {Handbook of cluster analysis.}}
\newblock {Chapman \& Hall/CRC}, Boca Raton, FL, 2016.

\bibitem{qrmdata}
M.~Hofert and K.~Hornik.
\newblock {\em {qrmdata: Data Sets for Quantitative Risk Management Practice}},
  2016.
\newblock R package version 2016-01-03-1.

\bibitem{Rcopula}
M.~Hofert, I.~Kojadinovic, M.~Maechler, and J.~Yan.
\newblock {\em copula: Multivariate Dependence with Copulas}, 2020.
\newblock R package version 0.999-20.

\bibitem{HubAra85}
L.~Hubert and P.~Arabie.
\newblock Comparing partitions.
\newblock {\em J. Classification}, 2:193--218, 1985.

\bibitem{Jietal18}
H.~Ji, H.~Wang, and B.~Liseo.
\newblock Portfolio diversification strategy via tail-dependence clustering and
  {ARMA-GARCH} vine copula approach.
\newblock {\em Aust Econ Pap}, 57(3):265--283, 2018.

\bibitem{Joe90}
H.~Joe.
\newblock Multivariate concordance.
\newblock {\em J. Multivariate Anal.}, 35(1):12--30, 1990.

\bibitem{Joe15}
H.~Joe.
\newblock {\em Dependence modeling with copulas}, volume 134 of {\em Monographs
  on Statistics and Applied Probability}.
\newblock CRC Press, Boca Raton, FL, 2015.

\bibitem{KocDes11}
I.~Koch and A.~De~Schepper.
\newblock Measuring comonotonicity in $m$--dimensional vectors.
\newblock {\em Astin {B}ulletin}, 41:191--213, 2011.

\bibitem{Koj04}
I.~Kojadinovic.
\newblock Agglomerative hierarchical clustering of continuous variables based
  on mutual information.
\newblock {\em Comput. Statist. Data Anal.}, 46:269--294, 2004.

\bibitem{Koj10}
I.~Kojadinovic.
\newblock Hierarchical clustering of continuous variables based on the
  empirical copula process and permutation linkages.
\newblock {\em Comput. Statist. Data Anal.}, 54(1):90--108, 2010.

\bibitem{KosKar16}
I.~Kosmidis and D.~Karlis.
\newblock Model-based clustering using copulas with applications.
\newblock {\em Stat. Comput.}, 26(5):1079--1099, 2016.

\bibitem{Kumar12}
S.~Kumar and N.~Deo.
\newblock Correlation and network analysis of global financial indices.
\newblock {\em Physical Review E - Statistical, Nonlinear, and Soft Matter
  Physics}, 86(2), 2012.

\bibitem{MaiSch12}
J.-F. Mai and M.~Scherer.
\newblock {\em Simulating copulas}.
\newblock Imperial College Press, London, 2012.

\bibitem{Mar17}
M.~Marbac, C.~Biernacki, and V.~Vandewalle.
\newblock Model-based clustering of {G}aussian copulas for mixed data.
\newblock {\em Comm. Statist. Theory Methods}, 46(23):11635--11656, 2017.

\bibitem{MueSca00}
A.~M\"{u}ller and M.~Scarsini.
\newblock Some remarks on the supermodular order.
\newblock {\em J. Multivariate Anal.}, 73(1):107--119, 2000.

\bibitem{MueSto02}
A.~M{\"u}ller and D.~Stoyan.
\newblock {\em Comparison methods for stochastic models and risks}.
\newblock Wiley Series in Probability and Statistics. John Wiley \& Sons Ltd.,
  Chichester, 2002.

\bibitem{Nel03}
R.~B. Nelsen.
\newblock Concordance and copulas: a survey.
\newblock In C.~M. Cuadras, J.~Fortiana, and J.~A. Rodr\'{i}guez-Lallena,
  editors, {\em Distributions with given marginals and Statistical Modelling},
  pages 169--178, Dordrecht, 2003. Kluwer.

\bibitem{Nel06}
R.~B. Nelsen.
\newblock {\em An Introduction to Copulas}.
\newblock Springer Series in Statistics. Springer, New York, second edition,
  2006.

\bibitem{Notterman2001}
D.~A. Notterman, U.~Alon, A.~J. Sierk, and A.~J. Levine.
\newblock Transcriptional gene expression profiles of colorectal adenoma,
  adenocarcinoma, and normal tissue examined by oligonucleotide arrays.
\newblock {\em Cancer Res.}, 61(7):3124--3130, 2001.

\bibitem{Pat12}
A.~J. Patton.
\newblock A review of copula models for economic time series.
\newblock {\em J. Multivariate Anal.}, 110:4--18, 2012.

\bibitem{Per19}
S.~Perreault, T.~Duchesne, and J.G. Ne\v{s}lehov\'a.
\newblock Detection of block-exchangeable structure in large-scale correlation
  matrices.
\newblock {\em J. Multivariate Anal.}, 169:400--422, 2019.

\bibitem{PucSca10}
G.~Puccetti and M.~Scarsini.
\newblock Multivariate comonotonicity.
\newblock {\em J. Multivariate Anal.}, 101(1):291--304, 2010.

\bibitem{PucWan15}
G.~Puccetti and R.~Wang.
\newblock Extremal dependence concepts.
\newblock {\em Statist. Sci.}, 30(4):485--517, 2015.

\bibitem{R}
{R Core Team}.
\newblock {\em R: A Language and Environment for Statistical Computing}.
\newblock R Foundation for Statistical Computing, Vienna, Austria, 2020.

\bibitem{Rand71}
W.~M. Rand.
\newblock Objective criteria for the evaluation of clustering methods.
\newblock {\em J. Amer. Statist. Assoc.}, 66(336):846--850, 1971.

\bibitem{Ross2000}
D.~T. Ross, U.~Scherf, M.~B. Eisen, C.~M. Perou, C.~Rees, P.~Spellman, V.~Iyer,
  S.~S. Jeffrey, M.~Van~de Rijn, M.~Waltham, A.~Pergamenschikov, J.~C. Lee,
  D.~Lashkari, D.~Shalon, T.~G. Myers, J.~N. Weinstein, D.~Botstein, and P.~O.
  Brown.
\newblock Systematic variation in gene expression patterns in human cancer cell
  lines.
\newblock {\em Nature Genetics}, 24:227--235, 2000.

\bibitem{Saletal07}
G.~Salvadori, C.~De~Michele, N.~T. Kottegoda, and R.~Rosso.
\newblock {\em Extremes in Nature. {A}n Approach Using Copulas}, volume~56 of
  {\em Water Science and Technology Library}.
\newblock Springer, Dordrecht (NL), 2007.

\bibitem{Sca84}
M.~Scarsini.
\newblock On measures of concordance.
\newblock {\em Stochastica}, 8(3):201--218, 1984.

\bibitem{Schmid07}
F.~Schmid and R.~Schmidt.
\newblock Multivariate conditional versions of spearman’s rho and related
  measures of tail dependence.
\newblock {\em Journal of Multivariate Analysis}, 98:1123–1140, 2007.

\bibitem{Sch_et_al10CTA}
F.~Schmid, R.~Schmidt, T.~Blumentritt, S.~Gaisser, and M.~Ruppert.
\newblock Copula-based measures of multivariate association.
\newblock In P.~Jaworski, F.~Durante, W.~K. H{\"a}rdle, and T.~Rychlik,
  editors, {\em {C}opula {T}heory and its {A}pplications}, volume 198 of {\em
  Lecture Notes in Statistics - Proceedings}, pages 209--236. Springer, Berlin
  Heidelberg, 2010.

\bibitem{Tay16}
M.~D. Taylor.
\newblock Multivariate measures of concordance for copulas and their marginals.
\newblock {\em Depend. Model.}, 4(1):224--236, 2016.

\bibitem{Ube05}
M.~{\'U}beda-Flores.
\newblock Multivariate versions of {B}lomqvist's beta and {S}pearman's
  footrule.
\newblock {\em Ann. Inst. Statist. Math.}, 57(4):781--788, 2005.

\bibitem{Yang2018}
C.~Yang, W.~Jiang, J.~Wu, X.~Liu, and Z.~Li.
\newblock Clustering of financial instruments using jump tail dependence
  coefficient.
\newblock {\em Stat. Methods Appl.}, 27(3):491--513, 2018.

\end{thebibliography}
}

\begin{table}[h]
\centering\caption{NCI60 data: ARI index of hierarchical clustering of low ($<0.3$) dependent tissues by varying dissimilarity measure and linkage method.}\label{tab:low}
\begin{tabular}{ccccc}
  \hline
                  & $d_\beta$ & $d_\phi$ & $d_\rho$ & $d_\tau$ \\
  \hline
  Average & \phantom{-}0.056  & 0.056 & 0.043 & 0.039 \\
  Single     & \phantom{-}0.056 & 0.056 & 0.056 & 0.056 \\
  Complete & -0.024                    & 0.005 & 0.039 & 0.039 \\
  \hline
\end{tabular}
\end{table}

\begin{table}[h]
\centering\caption{NCI60 data: ARI index of hierarchical clustering of high ($\geq 0.3$) dependent tissues by varying dissimilarity measure and linkage method.}\label{tab:high}
\begin{tabular}{ccccc}
  \hline
                  & $d_\beta$ & $d_\phi$ & $d_\rho$ & $d_\tau$ \\
  \hline
  Average   & 0.547 & 0.743 & 0.820 & 0.574 \\
  Single      & 0.116 & 0.076 & 0.076 & 0.298 \\
  Complete & 0.752 & 0.752 & 0.691 & 0.773 \\
   \hline
\end{tabular}
\end{table}

\begin{table}[h]
\centering\caption{Notterman's data: Rand index of hierarchical clustering results by varying dissimilarity measure and linkage method.}\label{tab:Not}
\begin{tabular}{ccccc}
  \hline
                  & $d_\beta$ & $d_\phi$ & $d_\rho$ & $d_\tau$ \\
  \hline
Average & 0.464 & 0.464 & 0.464 & 0.464 \\
Single & 0.464 & 0.464 & 0.464 & 0.464 \\
Complete & 1.000 & 0.464 & 0.571 & 0.464 \\
   \hline
\end{tabular}
\end{table}

\begin{table}[h]
\centering
\caption{Notterman's data: Steps of the hierarchical clustering procedure via dissimilarity function based on multivariate Kendall's tau.} \label{tab:Kendall}
\bigskip
\begin{tabular}{ccc}
\hline
merging variables
& $d_\tau$
& Kendall's tau
\\
\hline\hline
$5$ and $7$
& $0.061$
& $0.756$
\\
\hline\hline
$6$ and $8$
& $0.068$
& $0.728$
\\
\hline\hline
$1$ and $2$
& $0.076$
& $0.696$
\\
\hline\hline
$3$ and $4$
& $0.095$
& $0.620$
\\
\hline\hline
$(5,7)$ and $(6,8)$
& $0.124$
& $0.717$
\\
\hline\hline
$(1,2)$ and $(3,4)$
& $0.159$
& $0.637$
\\
\hline\hline
$(1,2,3,4)$ and $(5,6,7,8)$
& $0.209$
& $0.579$
\\
\hline
\end{tabular}
\end{table}

\begin{table}[h]
\centering\caption{S\&P 500: ARI index between the S\&P sector classification and the group composition provided by hierarchical clustering with different dissimilarity measures and linkage methods.}\label{tab:SP}
\begin{tabular}{lcccc}
  \hline
 & $d_\beta$ & $d_\phi$ & $d_\rho$ & $d_\tau$ \\
  \hline
   Average  & 0.003 & 0.006 & 0.003 & 0.003 \\
  Single      & 0.002 & 0.003 & 0.003 & 0.003 \\
  Complete & 0.331 & 0.337 & 0.320 & 0.370 \\
   \hline
\end{tabular}
\end{table}

\begin{table}[h]
\centering
\caption{Properties of dissimilarity functions.}\label{tab:diss-prop}
\bigskip
\begin{tabular}{ll}
(L1) & Order preserving property (lower orthant order) \\
(L1c) & Order preserving property (concordance order) \\
(L2) & Radially symmetry \\
(L3) & Continuity/Weak convergence \\
(G1) & Reducibility property \\
(G1s) & Strict reducibility property \\
(G2) & Comonotonic invariance\\
\end{tabular}
\end{table}

\begin{table}[h]
\centering
\caption{Summary of the properties satisfied (symbol: $\surd$), not satisfied (symbol: $\times$), or satisfied under specific conditions on $d^{1,1}$ (symbol: $\ast$) by the extended dissimilarity functions.}\label{tab:diss-satisfy}

\bigskip
\begin{tabular}{lccccccc}
\hline
& (L1) & (L1c) & (L2) & (L3) & (G1) & (G1s) & (G3) \\
\hline\hline
single linkage
& $\ast$ & $\ast$ & $\ast$ & $\ast$ & $\surd$  & $\times$ & $\surd$ \\
\hline\hline
average linkage
& $\ast$ & $\ast$ & $\ast$ & $\ast$ & $\surd$  & $\ast$  & $\ast$ \\
\hline\hline
complete linkage
& $\ast$ & $\ast$ & $\ast$ & $\ast$ & $\surd$  & $\ast$  & $\surd$ \\
\hline\hline
tail dependence
& $\surd$  & $\surd$  & $\times$ & $\times$ & $\surd$  & $\surd$  & $\surd$ \\
\hline\hline
Blomqvist's beta
& $\surd$  & $\surd$  & $\times$ & $\surd$  & $\surd$  & $\surd$  & $\surd$ \\
\hline\hline
Spearman's footrule
& $\surd$  & $\surd$  & $\times$ & $\surd$  & $\surd$  & $\surd$  & $\surd$ \\
\hline\hline
Kendall's tau
& $\times$ & $\surd$  & $\surd$  & $\surd$  & $\surd$  & $\surd$  & $\surd$ \\
\hline\hline
Spearman's rho
& $\surd$  & $\surd$  & $\times$ & $\surd$  & $\times$ & $\times$ & $\times$ \\
\hline
\end{tabular}
\end{table}

\begin{figure}
    \centering
  \includegraphics[width=.7\paperwidth,keepaspectratio=TRUE]{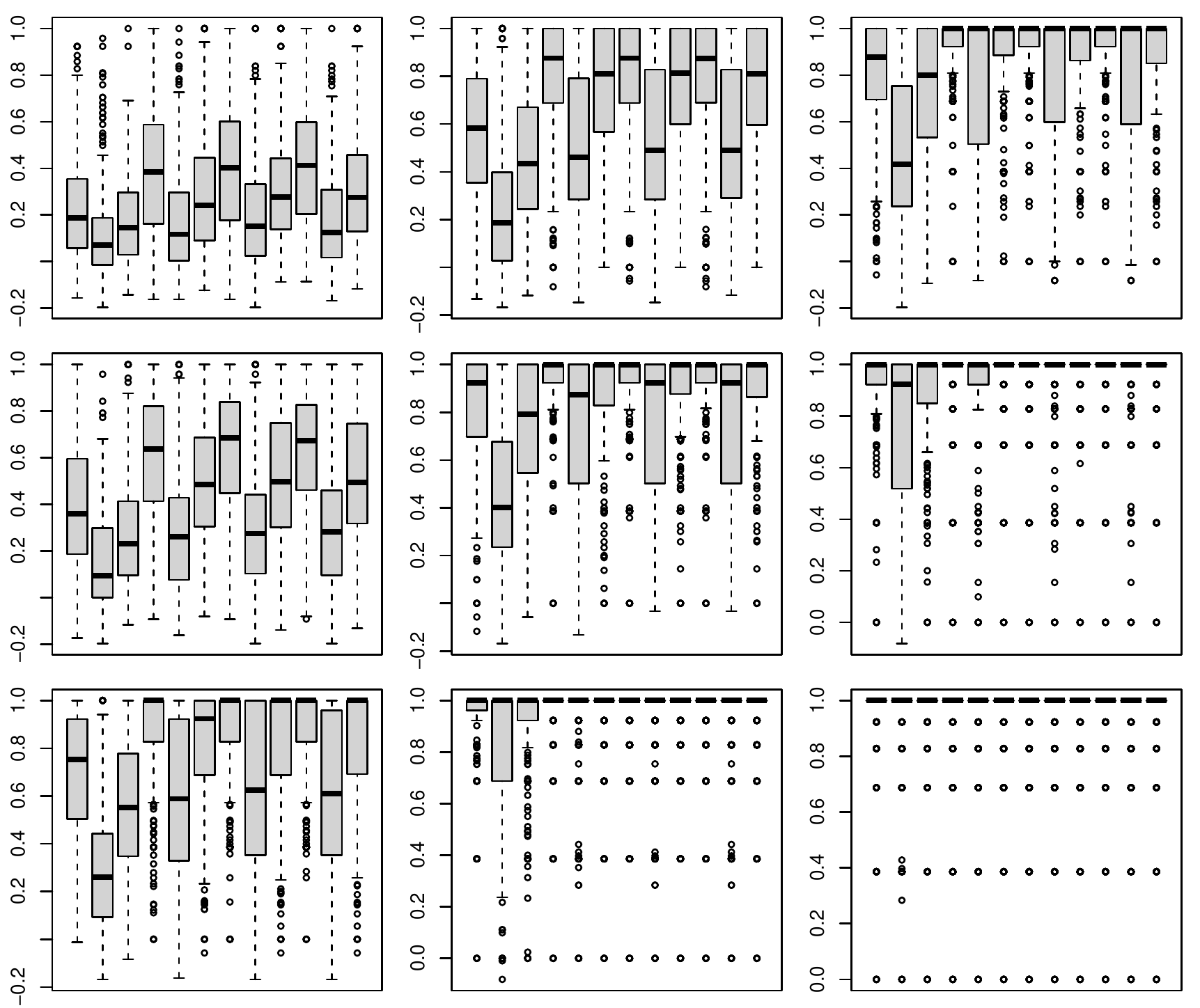}
  \caption{Boxplots of ARI (y-axis) by varying $i)$ pairwise dissimilarity measure among $d^{1,1}_\beta$, $d^{1,1}_\phi$, $d^{1,1}_{\rho}$ and $d^{1,1}_{\tau}$ and $ii$) linkage method among the average, single (minimum) and complete (maximum) one (x-axis starting with the average linkage and $d^{1,1}_\beta$, continuing with the single linkage and $d^{1,1}_\beta$ and ending with the complete linkage and $d^{1,1}_{\tau}$), $iii)$ Kendall's $\tau=(.1,.2,.3)$ (panels by cols), and $iv)$ sample size $N=50,100,250$ (panels by rows). Data simulated from independent groups with a Clayton copula within each group (see text).}\label{Clayton}
\end{figure}

\begin{figure}
  \centering
  \includegraphics[width=.7\paperwidth,keepaspectratio=TRUE]{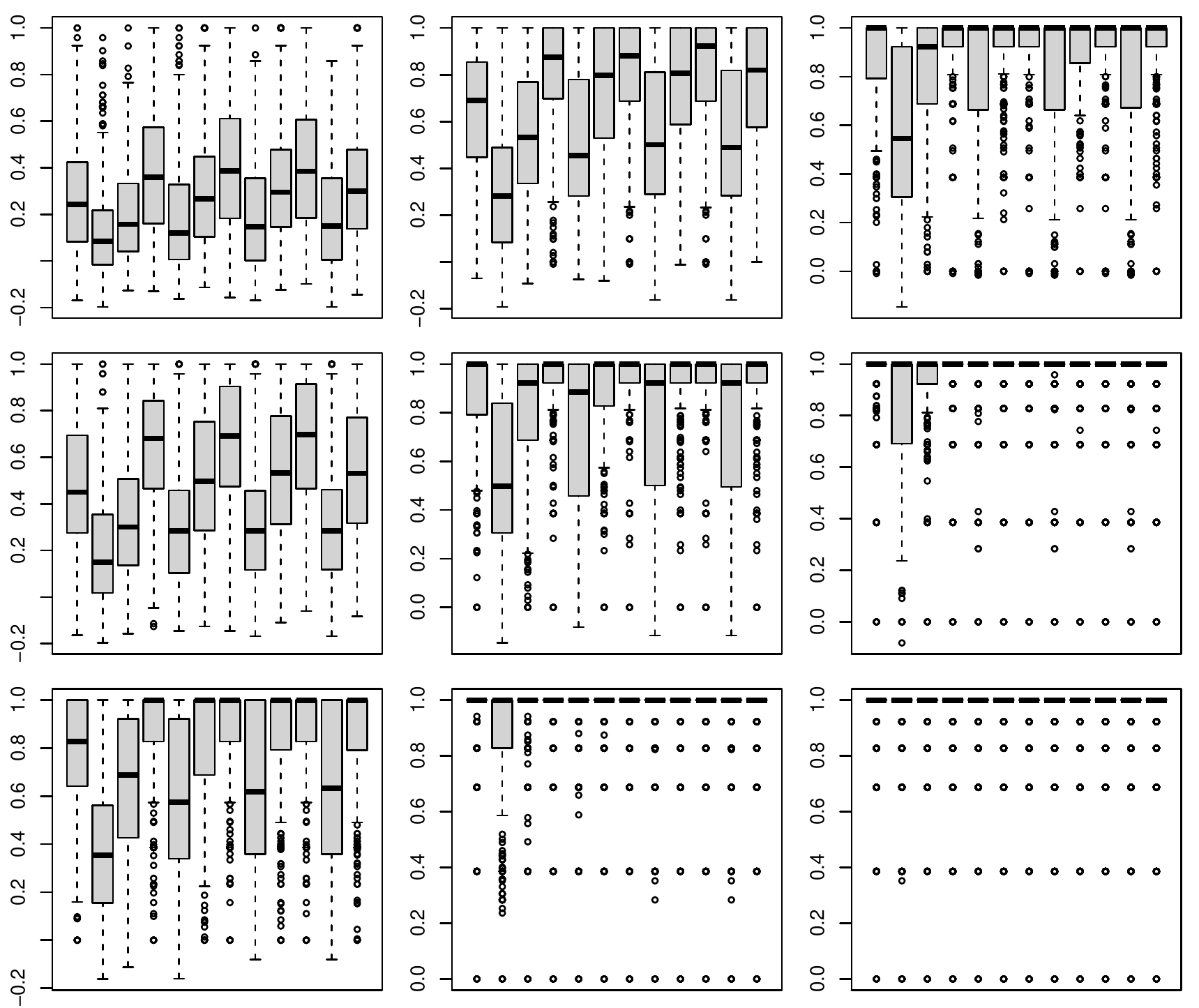}
  \caption{Boxplots of ARI (y-axis) by varying $i)$ pairwise dissimilarity measure among $d^{1,1}_\beta$, $d^{1,1}_\phi$, $d^{1,1}_{\rho}$ and $d^{1,1}_{\tau}$ and $ii$) linkage method among the average, single (minimum) and complete (maximum) one (x-axis starting with the average linkage and $d^{1,1}_\beta$, continuing with the single linkage and $d^{1,1}_\beta$ and ending with the complete linkage and $d^{1,1}_{\tau}$), $iii)$ Kendall's $\tau=(.1,.2,.3)$ (panels by cols), and $iv)$ sample size $N=50,100,250$ (panels by rows). Data simulated from independent groups with a Frank copula within each group (see text).}\label{Frank}
\end{figure}

\begin{figure}
  \centering
  \includegraphics[width=.7\paperwidth,keepaspectratio=TRUE]{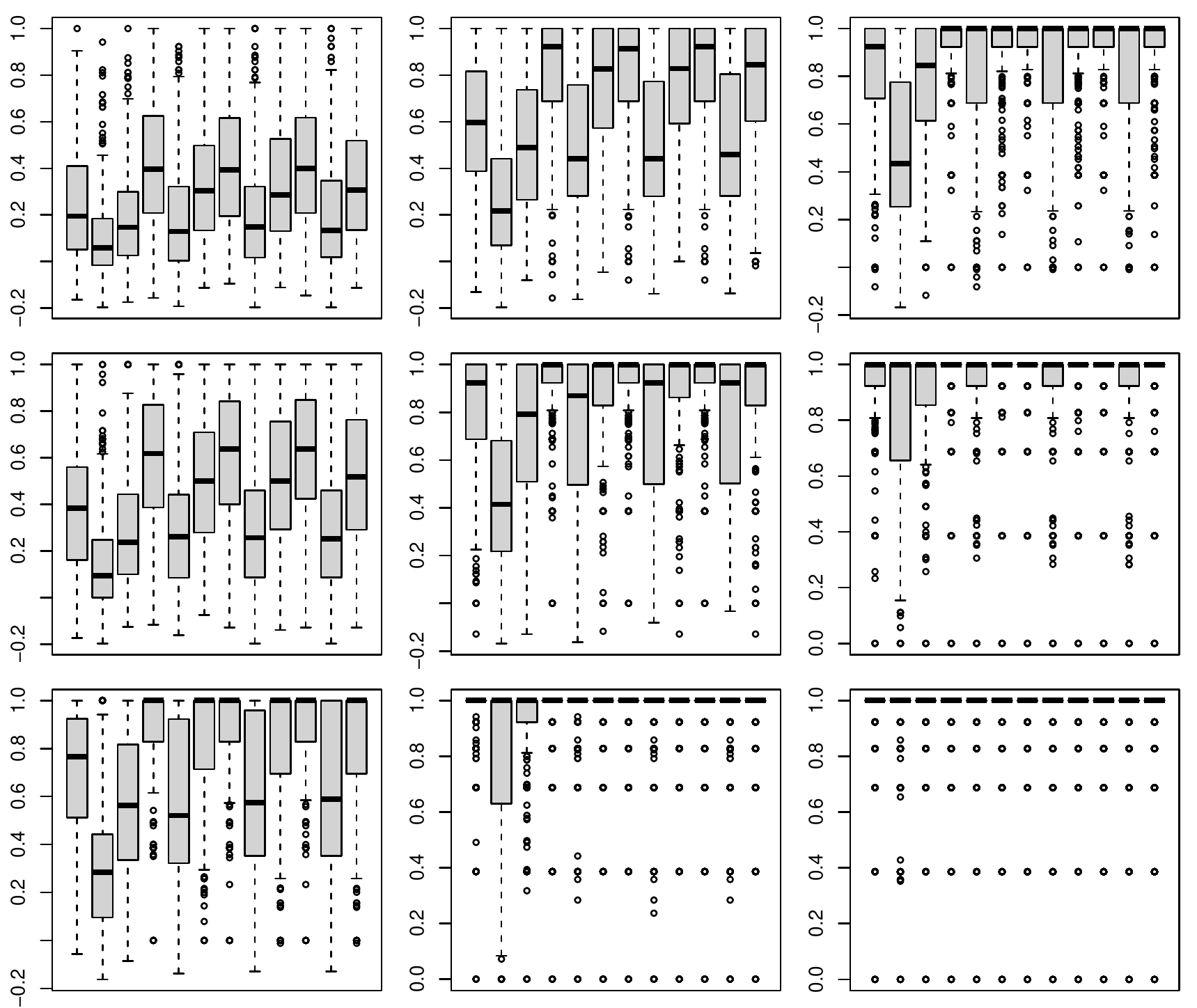}
  \caption{Boxplots of ARI (y-axis) by varying $i)$ pairwise dissimilarity measure among $d^{1,1}_\beta$, $d^{1,1}_\phi$, $d^{1,1}_{\rho}$ and $d^{1,1}_{\tau}$ and $ii$) linkage method among the average, single (minimum) and complete (maximum) one (x-axis starting with the average linkage and $d^{1,1}_\beta$, continuing with the single linkage and $d^{1,1}_\beta$ and ending with the complete linkage and $d^{1,1}_{\tau}$), $iii)$ Kendall's $\tau=(.1,.2,.3)$ (panels by cols), and $iv)$ sample size $N=50,100,250$ (panels by rows). Data simulated from independent groups with a Gumbel copula within each group (see text).}\label{Gumbel}
\end{figure}

\begin{figure}
  \centering
  \includegraphics[width=.7\paperwidth,keepaspectratio=TRUE]{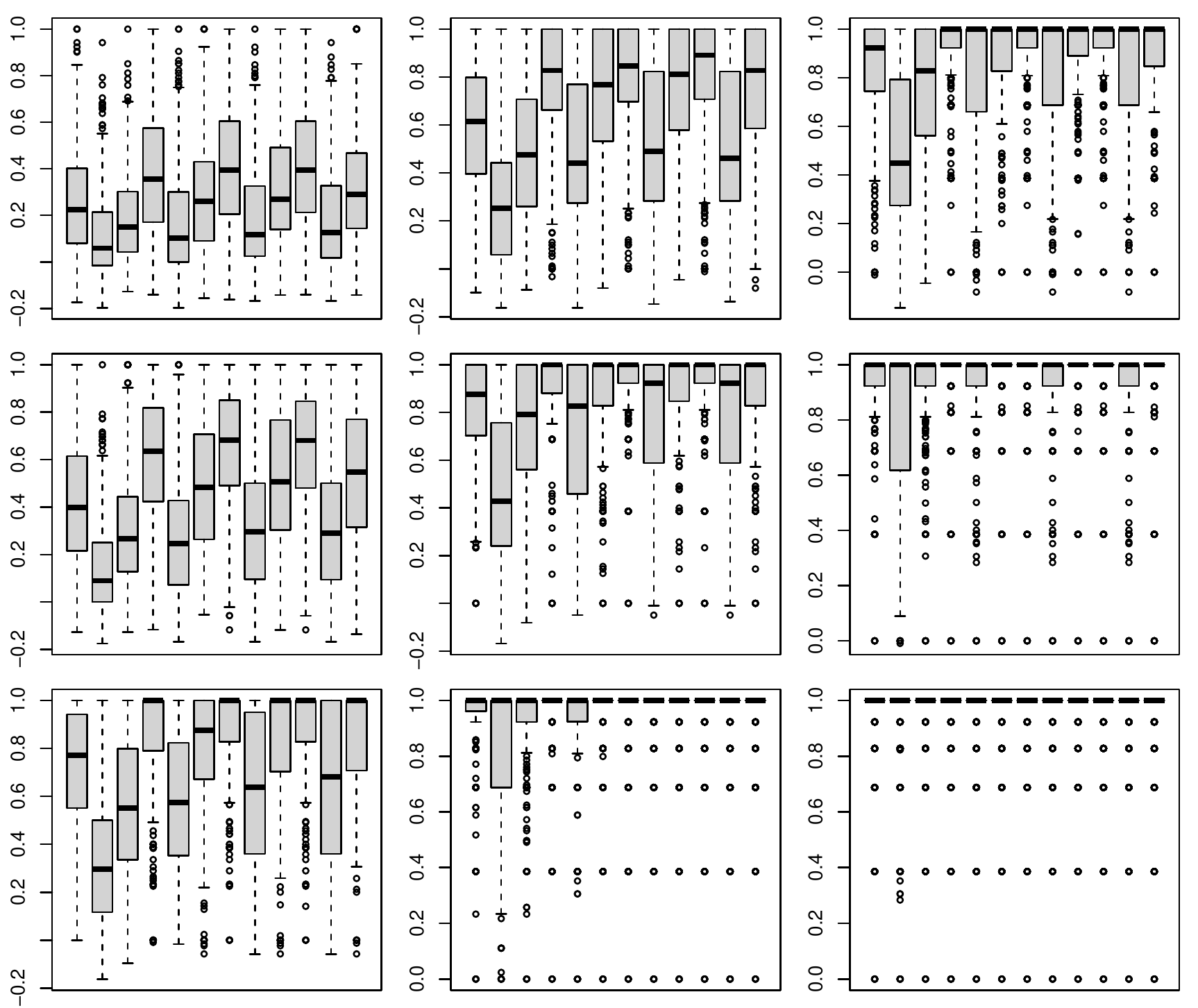}
  \caption{Boxplots of ARI (y-axis) by varying $i)$ pairwise dissimilarity measure among $d^{1,1}_\beta$, $d^{1,1}_\phi$, $d^{1,1}_{\rho}$ and $d^{1,1}_{\tau}$ and $ii$) linkage method among the average, single (minimum) and complete (maximum) one (x-axis starting with the average linkage and $d^{1,1}_\beta$, continuing with the single linkage and $d^{1,1}_\beta$ and ending with the complete linkage and $d^{1,1}_{\tau}$), $iii)$ Kendall's $\tau=(.1,.2,.3)$ (panels by cols), and $iv)$ sample size $N=50,100,250$ (panels by rows). Data simulated from independent groups with a equicorrelated Gaussian copula within each group (see text).}\label{Gaussian}
\end{figure}

\begin{figure}
  \centering
  \includegraphics[width=.7\paperwidth,keepaspectratio=TRUE]{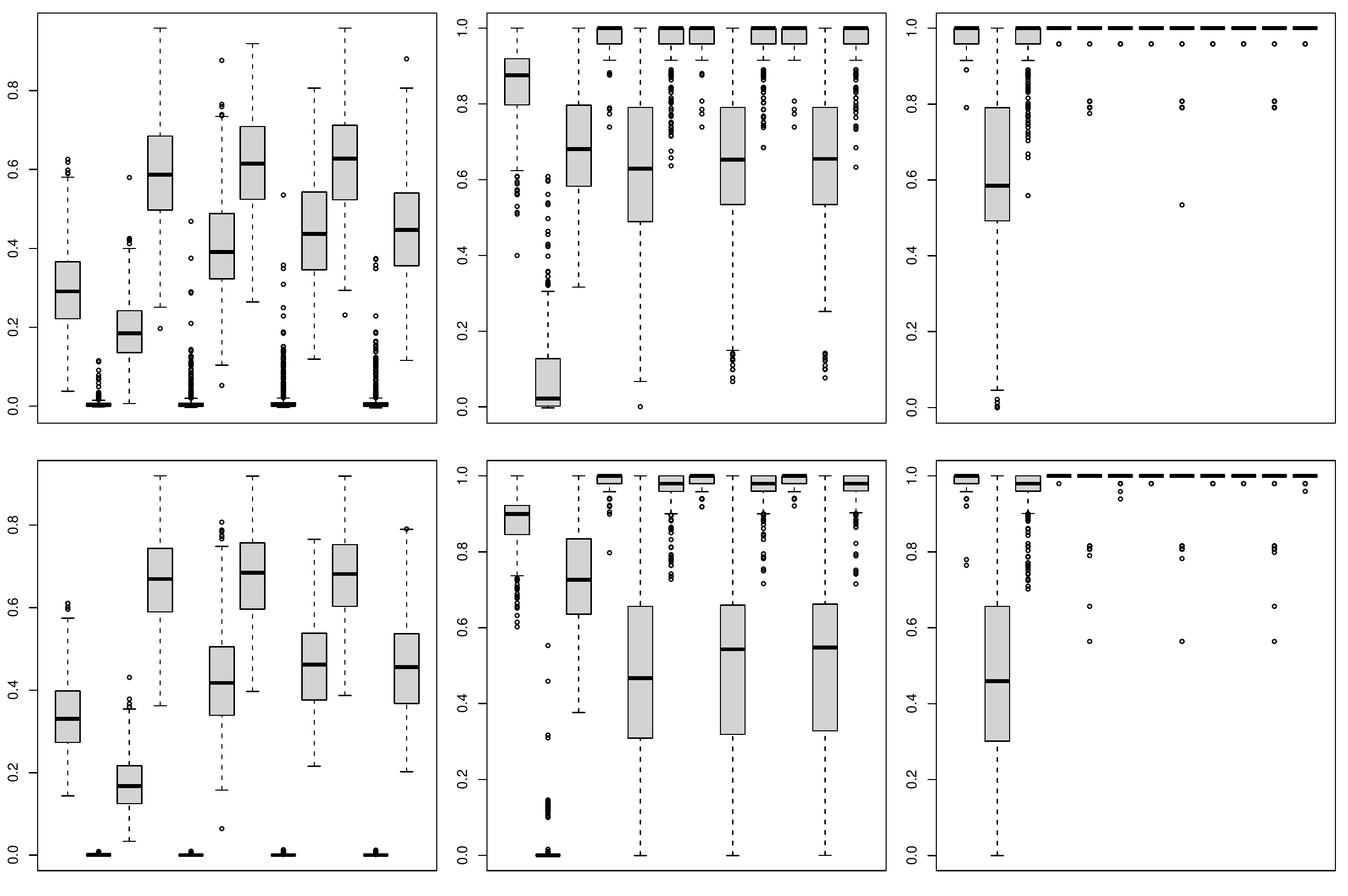}
  \caption{Boxplots of ARI (y-axis) by varying $i)$ pairwise dissimilarity measure among $d^{1,1}_\beta$, $d^{1,1}_\phi$, $d^{1,1}_{\rho}$ and $d^{1,1}_{\tau}$ and $ii$) linkage method among the average, single (minimum) and complete (maximum) one (x-axis starting with the average linkage and $d^{1,1}_\beta$, continuing with the single linkage and $d^{1,1}_\beta$ and ending with the complete linkage and $d^{1,1}_{\tau}$), $iii)$ Kendall's $\tau=(.1,.2,.3)$ (panels by cols), and $iv)$ clustering size (total number of variables) $m=60,120$ (panels by rows). Sample size is equal to $N=100$ and data are simulated from $K=6$ independent Clayton copulas of dimension $m/K$.}\label{100_3}
\end{figure}

\begin{figure}
  \centering
  \includegraphics[width=.7\paperwidth,keepaspectratio=TRUE]{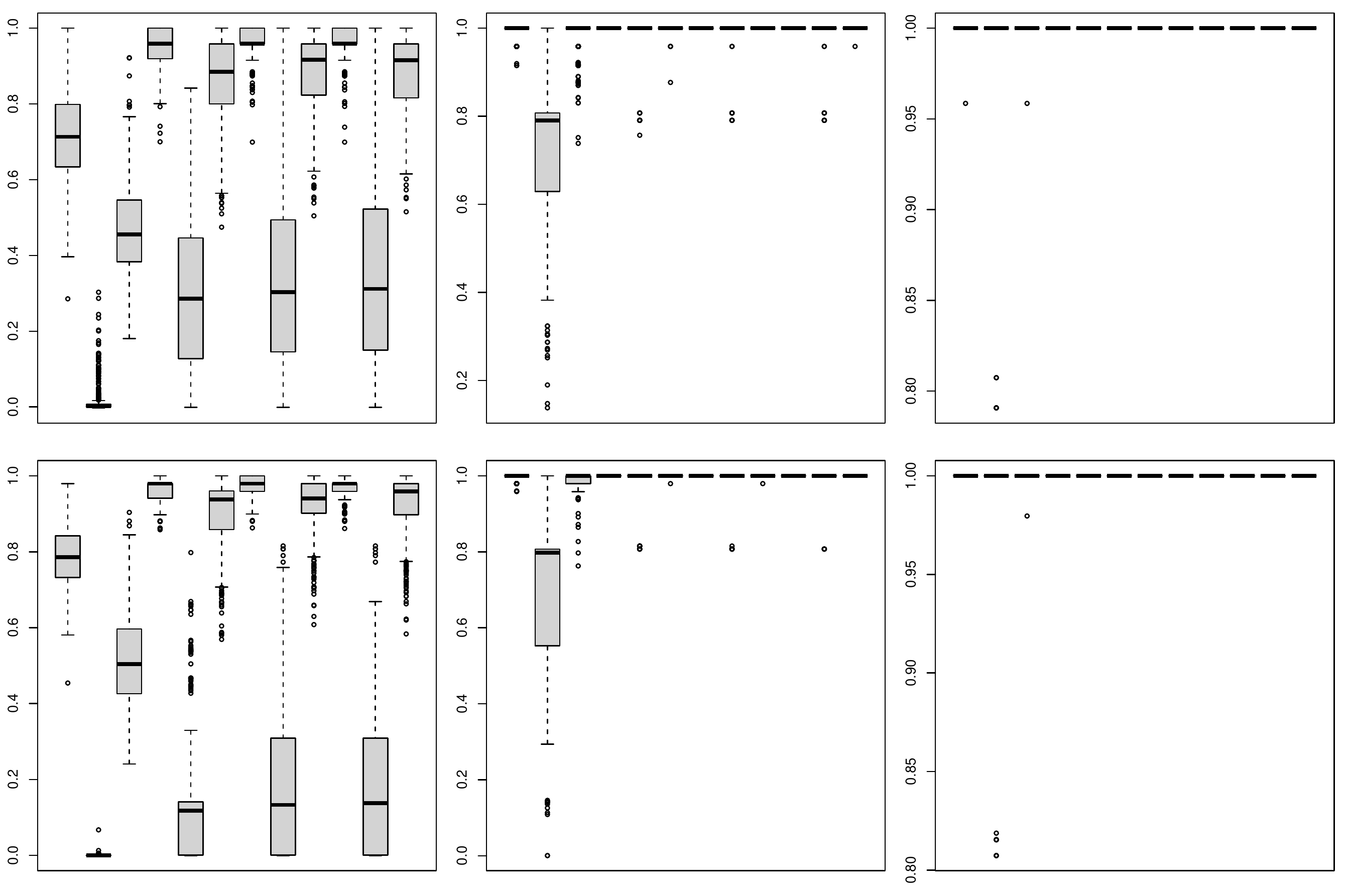}
  \caption{Boxplots of ARI (y-axis) by varying $i)$ pairwise dissimilarity measure among $d^{1,1}_\beta$, $d^{1,1}_\phi$, $d^{1,1}_{\rho}$ and $d^{1,1}_{\tau}$ and $ii$) linkage method among the average, single (minimum) and complete (maximum) one (x-axis starting with the average linkage and $d^{1,1}_\beta$, continuing with the single linkage and $d^{1,1}_\beta$ and ending with the complete linkage and $d^{1,1}_{\tau}$), $iii)$ Kendall's $\tau=(.1,.2,.3)$ (panels by cols), and $iv)$ clustering size (total number of variables) $m=60,120$ (panels by rows). Sample size is equal to $N=250$ and data are simulated from $K=6$ independent Clayton copulas of dimension $m/K$.}\label{250_3}
\end{figure}

\begin{figure}
  \centering
  \includegraphics[width=.7\paperwidth,keepaspectratio=TRUE]{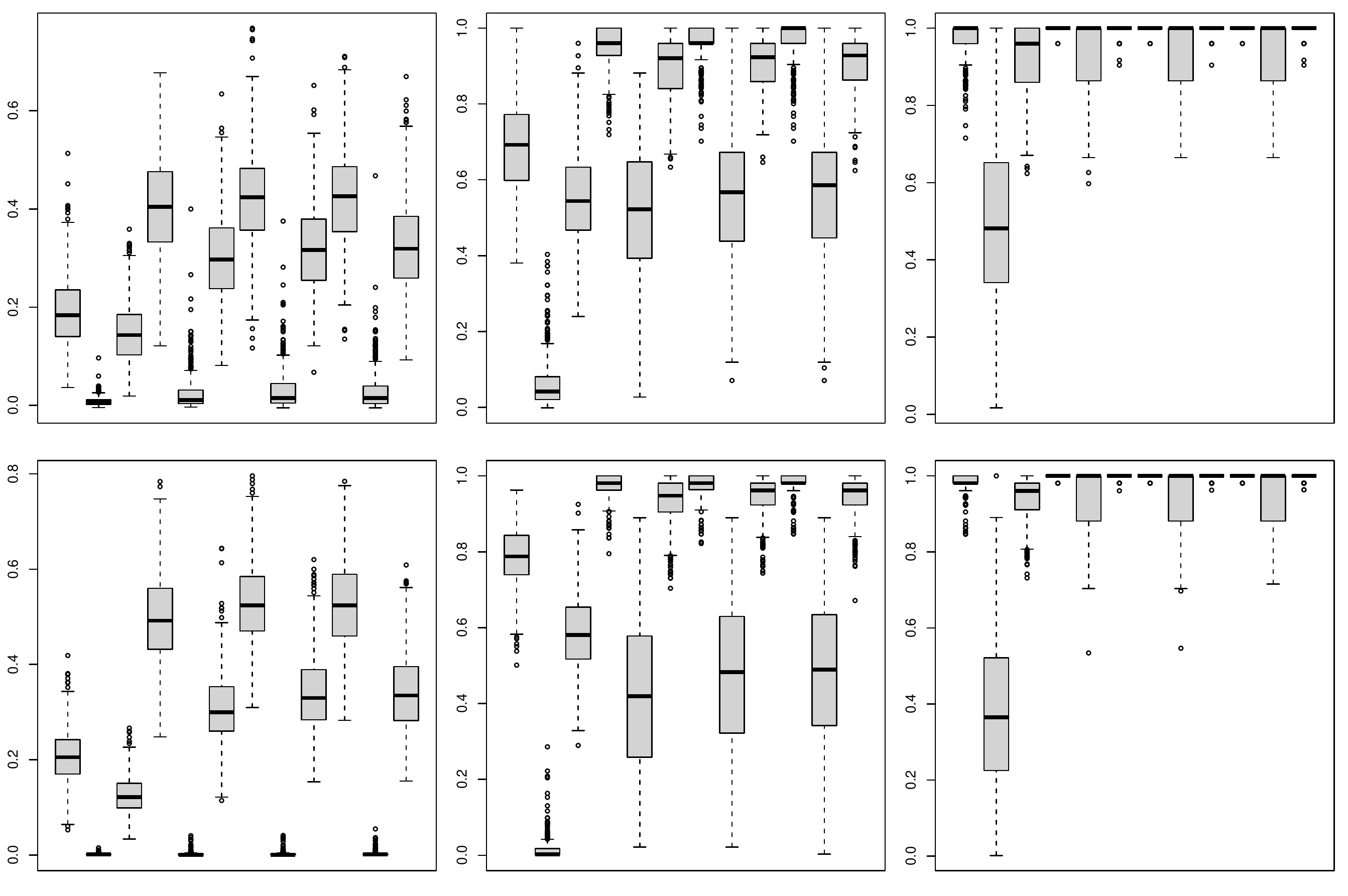}
  \caption{Boxplots of ARI (y-axis) by varying $i)$ pairwise dissimilarity measure among $d^{1,1}_\beta$, $d^{1,1}_\phi$, $d^{1,1}_{\rho}$ and $d^{1,1}_{\tau}$ and $ii$) linkage method among the average, single (minimum) and complete (maximum) one (x-axis starting with the average linkage and $d^{1,1}_\beta$, continuing with the single linkage and $d^{1,1}_\beta$ and ending with the complete linkage and $d^{1,1}_{\tau}$), $iii)$ Kendall's $\tau=(.1,.2,.3)$ (panels by cols), and $iv)$ clustering size (total number of variables) $m=60,120$ (panels by rows). Sample size is equal to $N=100$ and data are simulated from $K=10$ independent Clayton copulas of dimension $m/K$.}\label{100_6}
\end{figure}

\begin{figure}
  \centering
  \includegraphics[width=.7\paperwidth,keepaspectratio=TRUE]{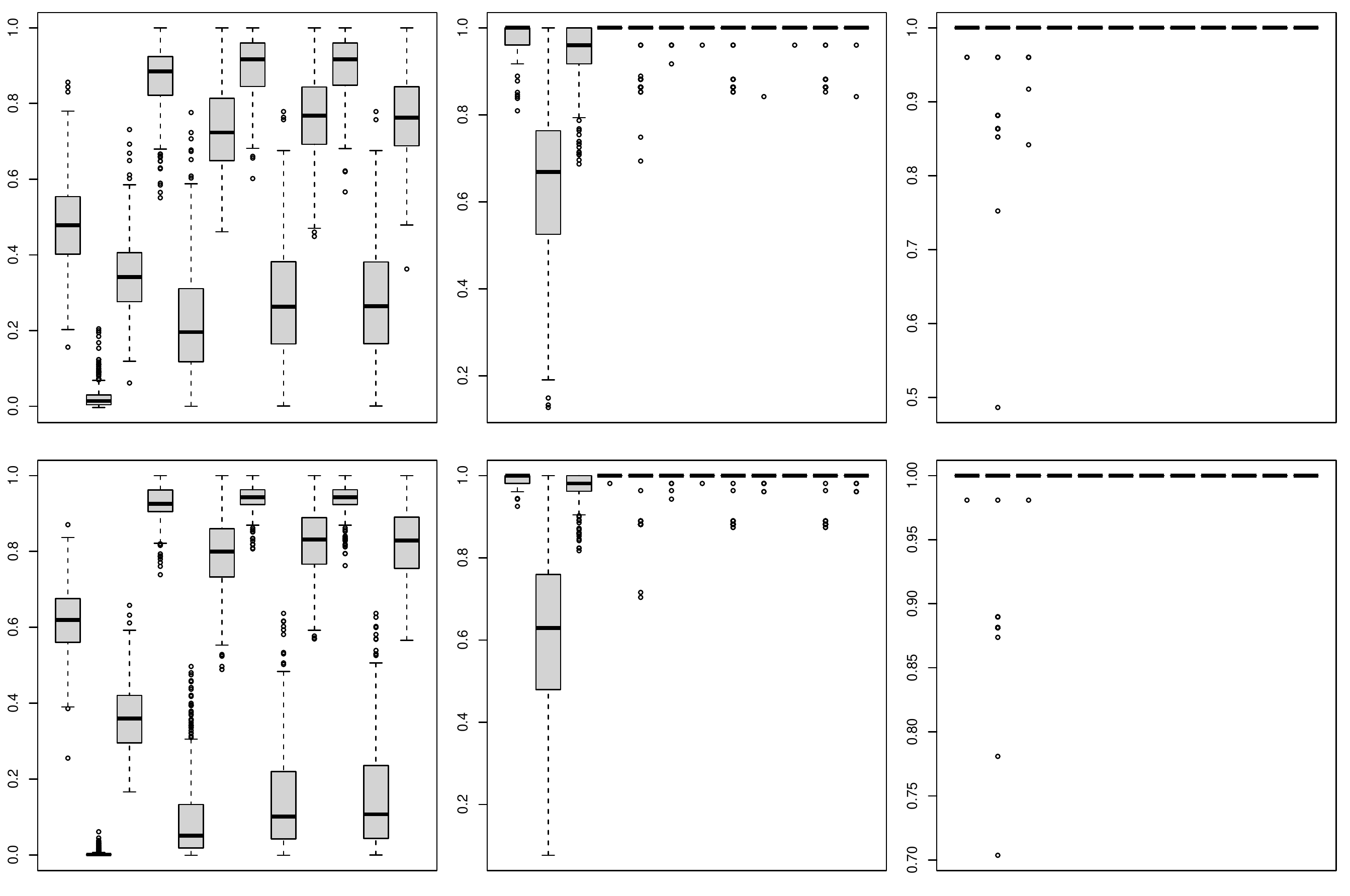}
  \caption{Boxplots of ARI (y-axis) by varying $i)$ pairwise dissimilarity measure among $d^{1,1}_\beta$, $d^{1,1}_\phi$, $d^{1,1}_{\rho}$ and $d^{1,1}_{\tau}$ and $ii$) linkage method among the average, single (minimum) and complete (maximum) one (x-axis starting with the average linkage and $d^{1,1}_\beta$, continuing with the single linkage and $d^{1,1}_\beta$ and ending with the complete linkage and $d^{1,1}_{\tau}$), $iii)$ Kendall's $\tau=(.1,.2,.3)$ (panels by cols), and $iv)$ clustering size (total number of variables) $m=60,120$ (panels by rows). Sample size is equal to $N=250$ and data are simulated from $K=10$ independent Clayton copulas of dimension $m/K$.}\label{250_6}
\end{figure}

\begin{figure}
  \centering
  \includegraphics[angle=270,width=.7\paperwidth,keepaspectratio=TRUE]{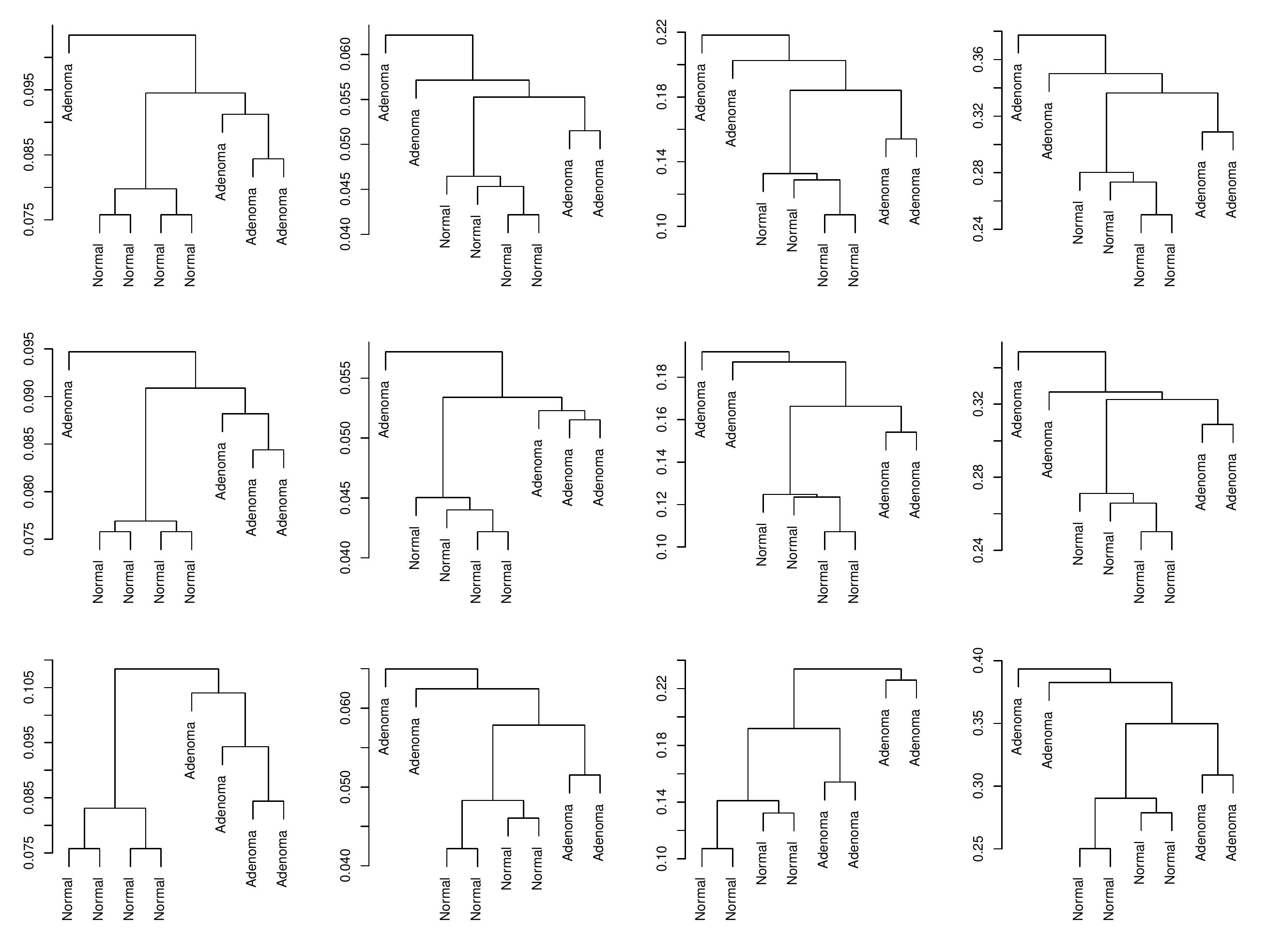}
  \caption{Dendrograms of the data set by \cite{Notterman2001}  by varying $i)$ dissimilarity measure among $d_\beta$, $d_\phi$, $1-\rho$ and $1-\tau$ by cols, and $ii$) linkage method among the average, single (minimum) and complete (maximum) one by row.}
	\label{fig:dendro}
\end{figure}
\end{document}